\documentclass[11pt]{article}
\usepackage{graphicx} 
\usepackage{amssymb}
\usepackage{bm}
\usepackage{rpmacros}
\usepackage{fullpage}
\usepackage{comment}
\usepackage{xcolor}
\usepackage[noend]{algpseudocode}
\usepackage{algorithm}
\usepackage{algorithmicx}
\usepackage{float}
\usepackage[left]{lineno}

\usepackage[colorlinks=true, linkcolor=blue, citecolor=blue, urlcolor=blue, pagebackref]{hyperref}

\hypersetup{
  linkcolor=[rgb]{0.3,0.3,0.6},
  citecolor=[rgb]{0.2, 0.6, 0.2},
  urlcolor=[rgb]{0.6, 0.2, 0.2}
}

\usepackage[nameinlink,capitalise]{cleveref}
\usepackage{nth}
\usepackage{intcalc}
\usepackage{etoolbox}
\usepackage{xstring}

\usepackage{ifpdf}
\ifpdf
\else
\usepackage[quadpoints=false]{hypdvips}
\fi

\newcommand{\ECCCurl}[2]{http://eccc.hpi-web.de/report/\ifnumcomp{#1}{>}{93}{19}{20}#1/#2/}

\newcommand{\shortECCC}[2]{\texttt{\href{http://eccc.hpi-web.de/report/\ifnumcomp{#1}{>}{93}{19}{20}#1/#2/}{eccc:TR#1-#2}}}

\newcommand{\parseECCC}[1]{
\StrSubstitute{#1}{TR}{}[\tmpstring]%
\IfSubStr{\tmpstring}{/}{ 
\StrBefore{\tmpstring}{/}[\ecccyear]%
\StrBehind{\tmpstring}{/}[\ecccreport]%
}{
\StrBefore{\tmpstring}{-}[\ecccyear]%
\StrBehind{\tmpstring}{-}[\ecccreport]%
}%
\shortECCC{\ecccyear}{\ecccreport}}

\title{Near-Quadratic Time Multivariate Polynomial Interpolation}
\author{Chandra Kanta Mohapatra\thanks{ Department of Computer Science. \texttt{chandra.mohapatra@aalto.fi}. The author is supported by Helsinki Institute for Information Technology (HIIT) postdoctoral fellowship. A part of this work was done when the author was a graduate student at Indian Institute of Technology, Bombay.} \\[0.15em]
    Aalto University}

\date{}

\usepackage{amsthm}
\usepackage{thmtools,thm-restate}

\numberwithin{equation}{section}
\declaretheoremstyle[bodyfont=\it,qed=\qedsymbol]{noproofstyle}

\declaretheorem[name=Observation,numbered=no]{observation*}

\declaretheorem[numberlike=equation]{theorem}

\declaretheorem[name=Theorem,numbered=no]{theorem*}

\declaretheorem[numberlike=equation]{lemma}
\declaretheorem[name=Lemma,numbered=no]{lemma*}

\declaretheorem[name=Corollary,numbered=no]{corollary*}

\declaretheorem[numberlike=equation]{proposition}
\declaretheorem[name=Proposition,numbered=no]{proposition*}

\declaretheorem[name=Claim,numbered=no]{claim*}

\declaretheorem[name=Conjecture,numbered=no]{conjecture*}

\declaretheorem[name=Question,numbered=no]{question*}

\declaretheoremstyle[bodyfont=\it,qed=$\lrcorner$]{defstyle}

\declaretheorem[unnumbered,name=Definition,style=defstyle]{definition*}

\declaretheorem[unnumbered,name=Example,style=defstyle]{example*}

\declaretheorem[unnumbered,name=Notation=defstyle]{notation*}

\declaretheorem[unnumbered,name=Construction,style=defstyle]{construction*}

\declaretheorem[numberlike=equation,style=defstyle]{remark}
\declaretheorem[unnumbered,name=Remark,style=defstyle]{remark*}



\newcommand{\vx}{\mathbf{x}}
\newcommand{\vy}{\mathbf{y}}

\newcommand{\vb}{\mathbf{{b}}}
\newcommand{\ve}{\textbf{e}}
\newcommand{\vv}{\mathbf{{v}}}
\newcommand{\vu}{\mathbf{{u}}}
\newcommand{\vw}{\mathbf{{w}}}
\newcommand{\krymin}[2]{\mu_{#2}^{#1}}
\newcommand{\matmin}[1]{\mu^{#1}}
\newcommand{\krylov}[2]{\mathcal{K}(#1,#2)}
\newcommand{\corank}{\text{co-rank}}

\newcommand{\vzero}{\bm{0}}
\newcommand{\intbr}[1]{\{0,1,\ldots,#1-1\}}

\begin{document}

\maketitle

\begin{abstract}
  
  Let $\F$ be a finite field. Given $N$ distinct points $\bm{\alpha}_1,\ldots,\bm{\alpha}_N \in \F^m$ and associated values $\beta_1,\ldots,\beta_N \in \F$, we study the problem of computing an $m$-variate polynomial $f(\vy) \in \F[\vy]$ of individual degree less than $d$, with $d^m \ge N$, that satisfies $f(\bm{\alpha}_i)=\beta_i$ for all $i\in[N]$, whenever such an $f$ exists. This is known as the\textit{ multivariate polynomial interpolation} problem. In this work, given $N=d^m$, we present a deterministic algorithm for multivariate polynomial interpolation over an arbitrary finite field $\F$ that runs in near-quadratic time $N^{2+o(1)}\poly(m,d,\log|\F|)$, assuming $d$ is $\omega(1)$.  This bound holds whenever the associated \emph{multivariate Vandermonde} matrix is nonsingular, has constant co-rank, or admits some special structure. In the general case, we give a randomized Las Vegas algorithm attaining the same near-quadratic running time, provided the finite field is sufficiently large and $d$ grows similar to the previous case. Unlike other multivariate interpolation strategies, we do not assume black-box access to the polynomial. In particular, we are not permitted to query arbitrary input points to obtain their evaluations. \medskip

Our interpolation result follows by combining Wiedemann's classical approach~\cite{w86} for solving sparse linear systems with known bounds for fast \emph{multivariate multipoint evaluation} over finite fields~\cite{umans08, KU11, BGKM23, BGGKU24}. Specifically, our deterministic interpolation strategy relies on a deterministic variant of Wiedemann's approach, which is more commonly used in the randomized setting.

\end{abstract}

\section{Introduction}

Solving linear systems is one of the most fundamental problems in computational algebra. Given a matrix $A\in \F^{n\times n}$ and a vector $\vb\in \F^n$, over any arbitrary finite field $\F$, the goal is to compute $\vx\in \F^n$ satisfying $A\vx=\vb$. The complexity of this problem is closely related to matrix multiplication, and the best known algorithms run in $O(n^\omega)$ time, where $\omega$ is the matrix multiplication exponent. In particular, no truly quadratic-time algorithm is known for arbitrary dense systems.

An important set of structured instances of linear system solving arises from polynomial interpolation. The input to the multivariate interpolation problem consists of $N \leq d^m\,$ $(d\,, m \in \N)$ distinct data pairs $\{(\bm\alpha_1,\beta_1),(\bm\alpha_2,\beta_2),\ldots,(\bm\alpha_N,\beta_N)\}$,  where $\bm\alpha_i\in \F^m$ and $\beta_i\in \F$ for each $i\in[N]$, and $\bm\alpha_i\neq \bm\alpha_j$ whenever $i\neq j$. The goal is to find an $m$-variate polynomial $f(\vy)\in\F[\vy]$ of individual degree less than $d$ such that $f(\bm\alpha_i)=\beta_i$ for all $i\in[N]$, whenever such a polynomial exists. The input of the mutivariate interpolation problem can be described using $N(m+1)$ field elements, hence that is the problem input size. Throughout this discussion, $N$ denotes the number of input points. The interpolation space for monomials of individual degree less than $d$ in $m$ variables has dimension $d^m$, thus for suitable $d,m \in \N$ such that $d^m \geq N$ ensures the interpolation space is large enough for $N$ input points.

For simplicity, we take $N = d^m$ throughout, so that the associated matrix, the \emph{multivariate Vandermonde matrix}\footnote{The rows of the matrix corresponds to evaluation points and the columns correspond to all possible monomials. Each entry is the evaluation of the corresponding monomial at the corresponding evaluation point.} $V$, is square. Squareness is needed to invoke Wiedemann's method
described in \autoref{sec:proof_overvw}. The case $N < d^m$ reduces to the square system by padding it with additional zero rows; see \autoref{rem:padding}. Throughout, the run time is always mentioned in terms of the number of underlying field operations. 



In the univariate setting ($m=1$), the input size is $2d$ field elements and, whenever the input points are distinct, there always exists a unique interpolating polynomial. This follows from the fact that the associated evaluation matrix, namely the Vandermonde matrix, has full rank at distinct points. Moreover, the structure of the Vandermonde matrix enables substantially faster algorithms than generic linear system solving. In particular, univariate interpolation admits near-linear time algorithms $\widetilde{O}(d)$ using fast multipoint evaluation due to Fiduccia \cite{Fid72} and, Borodin and Moenck \cite{BM74}. The multivariate setting is considerably more challenging. The associated interpolation matrix is a multivariate Vandermonde matrix, which may be singular even when all evaluation points are distinct. For example, consider $m=d=2$, a bivariate polynomial with individual degree less than $2$ has four distinct monomials $1,x,y,xy$. If we take distinct points $(x,y)$ such that $x=y$, then the $x$ and $y$ monomial columns become identical, leading to linear dependency. As a consequence, the interpolation problem may have multiple solutions or no solution at all.

The interpolation problem has been studied in various flavors. In one set up, the polynomial is given to us as a blackbox and we want to interpolate it using query access. That is, we can query points of our choice to obtain their evaluations. This set up is widely studied in the sparse polynomial interpolation literature such as Ben-Or--Tiwari~\cite{BenOrTiwari1988}. In contrast, dense multivariate interpolation naturally reduces to solving a large linear system. Thus the other set up is to consider structured evaluation points. That is, when the evaluation points come from a \textit{product-set} structure, interpolation can be performed in near-linear time in the input size using induction on each variable. However, as far as we are aware, no near-linear algorithm is known even for the general bivariate case over arbitrary evaluation points, which would suggest a route towards the general multivariate setting.

In the univariate case, fast interpolation~\cite[Algorithm 10.11]{gathenGerhard03} itself relies on fast multipoint evaluation as a subroutine, and both tasks admit near-linear time algorithms. The idea is to recursively solve the interpolation problem by dividing it into two roughly equal size sub-problems of the same kind. The recursion crucially relies on Euclidean division for univariates. A natural question is whether a similar reduction can be exploited in the multivariate setting. Unfortunately, it is not clear how to proceed with the recursion because of the non-Euclidean nature of multivariate polynomial rings. In this work we instead show how to exploit fast multivariate multipoint evaluation as a subroutine for multivariate interpolation using a different route.

Although the evaluation points themselves may be completely unstructured, the associated interpolation matrix remains a multivariate Vandermonde matrix, denoted by $V$. This naturally raises the question of whether one can exploit this algebraic structure to obtain algorithms faster than generic $O(N^\omega)$-time linear system solving. Turns out, if we can guarantee fast~(near-linear) matrix-vector multiplication then  using iterative linear algebra techniques solves the linear system faster compared to the trivial time bound.
For this strategy to work, we need to ensure multivariate Vandermonde matrix-vector multiplication, alternatively known as \emph{multivariate multipoint evaluation} runs fast. This problem has been studied in a series of work~\cite{umans08, KU11, BGKM23, BGGKU24} and known to have near-linear time algorithms over finite fields sometimes under some assumptions. The recent result of Bhargava, Ghosh, Guo, Kumar, and Umans~\cite{BGGKU24} showed that matrix-vector multiplication by multivariate Vandermonde matrices over all finite fields can be done in near-linear time, and we will mainly use their result in this work.

Since we have fast matrix-vector multiplication available, the next step is to employ fast iterative linear algebra techniques to solve the structured system. Our work uses Wiedemann's~\cite{w86} method for structured linear system solving. The Wiedemann framework basically does the following. Instead of explicitly inverting the matrix $A$ for solving $A\vx=\vb$, compute the minimal polynomial of the Krylov sequence 

\[
\krylov{A}{\vb}:=\{A^i\vb\}_{i\in\mathbb N},
\]
and use it to recover a solution vector $\vx$. The definitions of minimal polynomial and Krylov sequence can be found in the preliminaries~(\autoref{sec:prelim}). This is exactly what the deterministic variant of Wiedemann's method do.  Thus, the deterministic Wiedemann approach along with fast matrix-vector multiplication for generalized Vandermonde matrices enables us to get a near-quadratic time algorithm for multivariate interpolation over finite fields. The nature of our algorithm; that is \emph{deterministic} or \emph{randomized} depends on the rank of the system.  

Our results show that even for completely unstructured evaluation sets, the algebraic structure of generalized Vandermonde matrices helps to yield a nontrivial improvement over generic dense linear algebra and provides further evidence that multivariate interpolation may admit faster algorithms than previously known. We now state our main result.

\begin{theorem}[Informal]
Let $\F$ be a finite field and $N=d^m$. Given $N$ evaluation constraints
\[
f(\bm\alpha_i)=\beta_i,\qquad i\in[N],
\]
where $\bm\alpha_i\in\F^m$ and $\beta_i\in\F$, for an unknown $m$-variate polynomial $f(\vy)\in\F[\vy]$ of individual degree less than $d$, we can compute such a polynomial whenever one exists in time
$$N^{2+o(1)}\poly(m,d,\log|\F|).$$
The algorithm is deterministic when the associated multivariate Vandermonde matrix has full rank, constant co-rank  or admits additional structure\footnote{Formally we refer to the case in \autoref{prop: min-poly-constant-not-zero}.}, and otherwise by a Las Vegas randomized algorithm for the general case.
\end{theorem}

Our starting point is Wiedemann’s observation~\cite{w86} that solving a linear system can be reduced to computing the minimal polynomial of an associated Krylov sequence. This leads to an efficient framework whenever matrix-vector multiplication can be performed quickly, even in black-box setting, as is the case for sparse matrices. For nonsingular systems, the recovery of a solution from the Krylov sequence is relatively direct. The singular case is more non-trivial. Wiedemann uses random perturbations (by adding random sparse rows) to transform the system into a nonsingular one and then recovers a solution for the original system from the perturbed instance with high probability. Since the rank is not known apriori, he uses probabilistic rank computation as a subroutine. Kaltofen and Saunders~\cite{KS91} further refined this framework by introducing structured algebraic preconditioners based on Toeplitz matrices, leading to more efficient rank computations while requiring fewer random field elements. We crucially rely on Kaltofen and Saunders technique for handling highly singular systems that plays a crucial role in the general case of interpolation.

Note that most of the ideas we use in our work were already existing in the literature and we only develop the necessary tools to combine them for our purpose.

\section{Proof Overview}\label{sec:proof_overvw}
The main ingredient underlying our algorithm is Wiedemann's method for solving
sparse linear systems. Given a matrix $A \in \F^{n \times n}$ and a vector
$\vb \in \F^n$, Wiedemann's approach solves the system $A\vx = \vb$ when $A$ is
sparse. The technique exploits the fact that a sparse matrix-vector product can be computed more efficiently compared to the general one. When $A$ has $O(n)$ nonzero entries, matrix-vector product takes  $O(n)$ operations compared to quadratic time required in general. In fact, the technique extends beyond sparse systems: for
any matrix admitting a matrix-vector product faster than the trivial one,
Wiedemann's method can be employed to solve the associated system.

We first explain how to solve $A\vx = \vb$ in this approach,
assuming $A$ supports fast (say $\widetilde{O}(n)$) matrix-vector
multiplication. Suppose for some univariate polynomial
$p(x) = \sum_{i=0}^{d} p_i x^i \in \F[x]$ of degree $d$ we have $p(A) = 0$.
Multiplying the identity $p(A) = 0$ by $\vb$ gives
\[
  p(A)\vb \;=\; \sum_{i=0}^{d} p_i A^i \vb
  \;=\; p_0 \vb + p_1 A\vb + \cdots + p_d A^d \vb \;=\; \vzero,
\]
which, provided $p_0 \neq 0$, rearranges to
\[
  \vx \;=\; A^{-1}\vb
  \;=\; \frac{-1}{p_0}\left(p_1 \vb + p_2 A\vb + \cdots + p_d A^{d-1}\vb\right).
\]
Thus, once such a $p$ is known, the solution is obtained from $d$ matrix-vector
products with $A$. Such a polynomial always exists: by the Cayley--Hamilton
theorem~\cite[Section 25.5]{gathenGerhard03}, the \emph{characteristic polynomial} of $A$ annihilates $A$ and has degree at most $n$. So a nonzero annihilating polynomial of degree at most $n$ can always be used. Equivalently, the sequence of matrices $\{A^i\}_{i \geq 0}$ satisfies a linear recurrence of order at most $n$: every power $A^k$ is a linear combination of the $n$ preceding powers $A^{k-1}, \ldots, A^{k-n}$, with the coefficients of an annihilating polynomial used for the constants of the combination.

 Any annihilating polynomial may be used in place of $p$, and it is advantageous to use one of the smallest possible degree, since the degree determines the number of matrix-vector products. The unique lowest degree monic polynomial $\mu^A$ that also satisfies the relation $\mu^A(A) = 0$, called the \emph{minimal polynomial}. Since the above relation requires $p_0 \neq 0$, we need to ensure $\matmin{A}(0) \neq 0$. But, it is known that $\matmin{A}(0) \neq 0$ if and only if $A$ is nonsingular. Since the roots of $\matmin{A}$ are exactly the eigenvalues of $A$, we have $\matmin{A}(0)=0$ if and only if $0$ is an eigenvalue that is if and only if $A$ is singular. Thus, the above approach applies when $A$ is nonsingular, implying the constant term of the minimal polynomial to be nonzero, which is necessary for it to work. Then, we can simply work with $\mu^A$ instead $p$ for better runtime. For that, we need to explicitly compute the minimal polynomial $\mu^A$ which satisfies the sequence $\{A^i\}_{i \geq 0}$. It turns out, we can avoid computing $\mu^A$ altogether, rather it suffices to work with \emph{Krylov sequence} induced due to the original sequence $\{A^i\}_{i \geq 0}$, whose minimal polynomial already yields a solution for our system. 

The recurrence satisfied by $\{A^i\}_{i \geq 0}$ induces the vector sequence $\krylov{A}{\vb} := \{A^i \vb\}_{i \in \N}$, called the \emph{Krylov sequence} of $A$ at $\vb$. Any polynomial annihilating $\{A^i\}_{i \geq 0}$ also annihilates $\krylov{A}{\vb}$, for every $\vb \in \F^n$. The unique monic polynomial of least degree annihilating $\krylov{A}{\vb}$ is denoted $\krymin{A}{\vb}$ and is called the minimal polynomial of the Krylov sequence; it divides $\matmin{A}$, and its degree is therefore at most $n$. In particular, if $A$ is nonsingular then $\krymin{A}{\vb}$ has nonzero constant term, and the solution $\vx = A^{-1}\vb$ is recovered from $\krymin{A}{\vb}$ by the formula above.

Solving the system thus reduces to computing $\krymin{A}{\vb}$ efficiently. Note that the characteristic polynomial of $A$ would serve equally well, but computing it involves computing a determinant, which is not known to be possible in near-quadratic time for general matrices. Although in our set up the matrix $A$ is a multivariate Vandermonde matrix $V$, we are not aware of any deterministic
near-quadratic time algorithm, or any small closed-form expression for $\det(V)$ in the general case.\footnote{For example, closed-form expressions for $\det V$ can be derived on structured point sets, such as when $V$ is defined on a grid. Even in one variable, allowing general sparse exponents already makes the corresponding symbolic determinant hard: certain generalized Vandermonde determinants do not have polynomial-size algebraic formulas unless the algebraic complexity class VBP = VF \cite{CKLMSS}.} Moreover, the characteristic polynomial need not have the least degree, and using it may lead to more matrix-vector products than necessary. Fast matrix-vector multiplication will be useful in improving the overall runtime and more crucially to compute the minimal polynomial efficiently, as we explain below.

For a linear recurrence over \emph{scalars}, the minimal polynomial
is recovered by the Berlekamp--Massey algorithm~\cite[Algorithm 12.9]{gathenGerhard03}: given an upper bound on the degree of the recurrence together with sufficiently many of its initial terms, it outputs the minimal polynomial in near-linear time by solving a structure linear system. The sequence $\krylov{A}{\vb}$, however, is \emph{vector-valued} rather than \emph{scalar-valued}. We describe next how Wiedemann's method reduces the vector case to the scalar one and what we do instead in order to obtain a deterministic algorithm, and how the singular case is handled eventually.

\paragraph*{Nonsingular system.}
Wiedemann's randomized approach converts $\krylov{A}{\vb}$ into the scalar
sequence $\{\vu^\top A^i \vb\}_{i \geq 0}$ for a random vector
$\vu \in \F^n$; with high probability the minimal polynomial of this scalar sequence coincides with $\krymin{A}{\vb}$. In this work we use a deterministic variant instead. For each $j \in [n]$, consider the coordinate sequence $s_j(i) = \ve_j^\top A^i \vb$ and let $\mu_j$ denote its minimal polynomial.
A polynomial annihilates $\krylov{A}{\vb}$ if and only if it annihilates every coordinate sequence, whence
\[
  \krymin{A}{\vb}(x) \;=\; \operatorname{lcm}\left(\mu_1(x), \ldots, \mu_n(x)\right).
\]
Accordingly, after computing the vectors $\{A^i \vb\}_{i=0}^{2n-1}$, we apply Berlekamp--Massey independently to each of the $n$ coordinate sequences and recover $\krymin{A}{\vb}$ deterministically in near-quadratic time.

In our setting, the matrix $A$ is a multivariate Vandermonde matrix arising from multivariate interpolation. A recent result of Bhargava et
al.~\cite{BGGKU24}, the culmination of a series of works on multivariate
multipoint evaluation~\cite{umans08, KU11, BGKM23}, shows that matrix-vector multiplication by such matrices can be performed in near-linear time over every finite field. Wiedemann's framework therefore yields a near-quadratic algorithm for multivariate interpolation.

\paragraph*{Singular system.}

We now revisit Wiedemann's method for the case when the matrix $A$ is not full rank. Since $A^{-1}$ may not exist, the standard derivation does not apply directly. Moreover, a first order problem seems like detecting the non-invertibility of the system, which is already computationally equivalent to computing the determinant of the matrix, that may take \(n^{\omega}\). But, it turns out that we can handle the non-invertible case with some refined ideas. Recall that the constant term of the characteristic polynomial of $A$ is, up to sign, equal to $\det(A)$. Since the minimal polynomial of the Krylov sequence $\{A^i\vb\}_{i\ge 0}$, denoted by $\krymin{A}{\vb}$, divides the characteristic polynomial of $A$, it follows that $\det(A)\neq 0$ implies $\krymin{A}{\vb}(0)\neq 0$. Thus, we first compute $\krymin{A}{\vb}$ and examine its constant term. The converse, however, need not hold: it is possible for $A$ to be singular while $\krymin{A}{\vb}(0)\neq 0$. To handle this scenario, we observe that whenever the minimal polynomial of the Krylov sequence has a nonzero constant term, the vector \(\vb\) lies in the column space of \(A\). Consequently, computing a solution using the minimal polynomial of Krylov sequence always yields a valid solution (possibly one among many), even when the matrix is not invertible. Since the resulting vector can be verified directly, no additional certificate of invertibility is required. Therefore, whenever $\krymin{A}{\vb}(0)\neq 0$, the standard Wiedemann construction yields a valid solution to $A\vx=\vb$, even if $A$ is singular.

Thus, the singular case becomes more non-trivial when $\krymin{A}{\vb}(0)=0$. In this case, it is possible to show that, the minimal polynomial yields a nontrivial linear dependence among the columns of $A$, allowing us to remove one redundant column. Removing a column produces a rectangular system assuming we started with a square system, and to identify a redundant row, we apply the same procedure recursively to the transpose of the reduced system. Either the resulting subsystem is solvable which implies that the corresponding row is redundant, or the associated Krylov minimal polynomial again has a vanishing constant term that yields another dependence relation. Repeating this process recursively removes one redundant row and one redundant column at each step.

When the co-rank is constant, only constantly many recursive steps are needed to obtain a square subsystem that has full rank, which can be solved directly using the nonsingular Wiedemann framework.  Solving this subsystem gives a solution to the original system after restoring the free variables. For larger co-rank, we use a different approach from \cite{KS91}. We use randomized preconditioning by triangular Toeplitz matrices, so that the leading principal minors remain nonsingular up to the rank of the matrix, with high probability. This reduces the singular system to solving a full-rank principal subsystem along with the final consistency check. We show that the overall approach can be done in near-quadratic time.

Note that, throughout the procedure, no matrix entries are modified. We only remove redundant rows and columns in order to obtain a nonsingular subsystem.

\section{Related Work}

There is a large body of work on polynomial interpolation, as it is a fundamental algebraic problem that arises throughout computational mathematics, coding theory, algebraic complexity, symbolic computation, approximation theory, and adjacent areas. The problem has been studied extensively by both mathematicians and computer scientists~\cite{ChungYao1977,deBoorRon1990,GS00}. The works on multivariate interpolation can roughly be divided into several broad paradigms, depending on the structure of the evaluation points, the polynomial being recovered, or the way the polynomial is accessed.

\medskip

\noindent\textbf{Structured evaluation points.}
When the evaluation points have a product set (grid) structure, multivariate interpolation reduces recursively to a sequence of univariate interpolation problems~\cite{PAN1994}. Since the univariate interpolation at $d$ many points can be performed in $\widetilde{O}(d)$ arithmetic operations using fast polynomial arithmetic \cite{BM74}, multivariate interpolation on an $m$-dimensional grid of size $d^m$ also admits a near-linear time $\widetilde{O}{(d^m)}$ algorithm using induction on the number of variables. More generally, when the evaluation points form a downward-closed (lower) set of a tensor-product grid, interpolation spaces admit structured Newton-type bases and divided-difference constructions~\cite{DynFloater2014, N19}. Van der Hoeven and Schost~\cite{vanderHoevenSchost2013} obtained quasi-linear time
interpolation in this setting, and recent works have developed further fast algorithms~\cite{Hofmann2025, FKW25}.

A common framework of all these works is that the geometry of the evaluation points is known in advance. In contrast, we make no assumption whatsoever on the structure of the input points.

\medskip
\noindent\textbf{Interpolation with black-box access.}
Another common approach treats the polynomial as a black box that the algorithm can evaluate at points of its choice. In the sparse setting, the goal is to recover the polynomial's support and coefficients by accessing the black box evaluations. Zippel~\cite{Zippel1979,Zippel1990} gave probabilistic algorithms in this model. Ben-Or and Tiwari~\cite{BenOrTiwari1988} gave a deterministic sparse interpolation algorithm which relies on Prony's technique using linear recurrences. Subsequent works also exploited Prony's method~\cite{GiebrechtLabahnLee2002}, improved complexity using fast Toeplitz and transposed Vandermonde systems~\cite{KaltofenYagati1988}, and further removed apriori bounds on the degree and sparsity~\cite{KaltofenLee2003} using early termination. Numerous variants of sparse polynomial interpolation based on Kronecker substitution, derivative evaluations, finite coefficient sets, and arithmetic circuit representations~\cite{GrigorievKarpinskiSinger1990,GrigorievKarpinski1993,Werther1994,Man95,GargSchost09,BlaserJindal2014,ArnoldRoche14,ArnoldGiesroche16,HuangGao2017,HuangGao2020,GiorgiGrePerrRoc22,Huang2023,Grenet2023,HoevenLec25} have also been studied in the literature. Some of these techniques also apply to dense interpolation, for example the framework of Zippel~\cite{Zippel1990}. However, a common feature of all these works is that the algorithm is free to choose evaluation points, unlike in our setting.

\medskip
\noindent\textbf{The fixed-point interpolation setting.} In our interpolation set up, we neither assume evaluation query access nor any geometric assumption on the input points. In the sparse univariate setting with prescribed points, Borodin and
Tiwari~\cite{BorodinTiwari1991} studied whether a given set of points admits a
sparse interpolant. The most closely related work in this setting is that of Zilić and Vranešić~\cite{ZilicVranesic2002}. They construct an interpolation basis consisting of $t$ monomials whose associated evaluation matrix is nonsingular, by iteratively replacing monomials guided by the nullspace structure of the multivariate Vandermonde matrix. Their algorithm is deterministic, works over arbitrary fields, and requires $O(t^{\omega+1})$ arithmetic operations using fast matrix multiplication. In contrast, we fix in advance the interpolation space consisting of all monomials of individual degree less than $d$  and study the resulting multivariate Vandermonde system directly, including the case where it is rank deficient. 

Multivariate interpolation has also been studied with derivative
conditions~\cite{SauerXu1995b, Lorentz1992} and via vanishing ideals and Gr\"obner bases~\cite{MollerBuchberger1982, MarinariMollerMora1993}. Their setting is different, so we do not discuss them further.

We also note that there is a deterministic version for solving invertible systems in Wiedemann's original result \cite{w86}~(Algorithm $2$) which is very close to what we do and rediscovered independently during the course of this work. In our approach, we apply the Berlekamp-Massey for each coordinate in parallel unlike the Wiedemann algorithm, but nevertheless both are in similar spirit. We suspect Wiedemann's deterministic approach is not commonly used because the space complexity is not linear; rather, it becomes quadratic in the input size, and the whole black-box setup loses much of its appeal, since in quadratic space we can store the entire input matrix as a whitebox by figuring out all its columns, which defeats its black-box nature of system solving. In our work, we do not stress the black-box setting since we know the monomial basis and the input points, and can explicitly determine each entry of the input system. Since we are not in the black-box framework, we can trade quadratic space for determinism. A recent work of Liu, Nguyen and Yang \cite{LiuNguyenYang26} discusses the time and space tradeoff in the setting of rational general system solving.

To the best of our knowledge, we are not aware of any prior work obtaining near-quadratic algorithms for dense multivariate interpolation over finite fields at arbitrary fixed evaluation points without assuming sparsity, black-box access, or special structure on the input points.

\section{Organization}

\autoref{sec:mat-vec-complexity} and \autoref{subsec:transposed-mme} discusses the complexity of computing matrix-vector product and the transposed matrix-vector product respectively when the matrix is a multivariate Vandermonde. \autoref{sec:det-krylov} performs the deterministic Wiedemann's computation and then the interpolation system solving is addressed in \autoref{sec:nonsingular} and \autoref{sec:singular} when the system is nonsingular and singular respectively.


\section{Preliminaries and Notation}\label{sec:prelim}
Throughout the paper, let \(\F\) denote a finite field, $\mathbb{F}^*= \F\setminus \{0\}$, let \(\N\) denote the set of natural numbers including $0$, and let \([n]:=\{1,2,\dots,n\}\). The soft-\(O\) notation \(\widetilde{O}(m)\) denotes \(O(m\log^c m)\), where \(c\) is a constant.

\paragraph{Linear recurrences.}
A sequence $(a_i)_{i\ge 0}$ over $\F$ satisfies a linear recurrence of order $d$ if there exists an with \emph{annihilating polynomial} $c(x)=x^d+c_{d-1}x^{d-1}+\cdots+c_0\in\F[x]$ such that $\sum_{j=0}^{d}c_j a_{i+j}=0$ for all $i\ge 0$ (with $c_d=1$). The unique monic annihilating polynomial of minimum degree is called the \emph{minimal polynomial} of the sequence.

The above definition applies verbatim to the setting where the sequence elements are not just field elements in $\F$ rather it can be vectors or matrices i.e. sequences in $\F^n$ or $\F^{n\times n}$. For the matrix sequence $a_i=A^i$, $A\in \F^{n \times n}$, this is equivalent to $c(A)=\bm0$, and for the Krylov-sequence where $a_i={A^i\vv}$ for $\vv \in \F^n$, this is equivalent to $c(A)\vv=\bm0 \in \F^n$. 

\paragraph{Matrix polynomials.}
For $A\in\F^{n\times n}$, let $\chi_A(x):=\det(xI-A)$ denote the characteristic polynomial and $\matmin{A}(x)$ the minimal polynomial, i.e., the unique monic polynomial of least degree satisfying $\matmin{A}(A)=0$. By the Cayley--Hamilton theorem, $\matmin{A}(x)\mid\chi_A(x)$.

\paragraph{Krylov sequences.}
For $A\in\F^{n\times n}$ and $\vv\in\F^n$, the \emph{Krylov sequence} $\krylov{A}{\vv}$ is defined as a sequence of vectors in repeated matrix-vector product form: $$\krylov{A}{\vv}:=(\vv, A\vv, A^2\vv, \ldots\ldots)=\{A^i\vv\}_{i\in\mathbb{N}}.$$

The unique monic annihilating polynomial $\krymin{A}{\vv}(x)$ of least degree that satisfies $\krymin{A}{\vv}(A)\vv=\bm0$ is called the minimal polynomial of the $\krylov{A}{\vv}$. Since $\matmin{A}(A)=\bm0$, we have $\krymin{A}{\vv}(x)\mid\matmin{A}(x)$. 


\paragraph{Berlekamp--Massey.}

The Berlekamp--Massey algorithm~(Chapter~12, \cite{gathenGerhard03}) takes as input the first \(2d\) terms of a scalar linear recurrence of degree at most \(d\) and computes its minimal polynomial in \(O(d\log^2 d\log\log d)=\widetilde{O}(d)\) field operations.

\medskip
A concrete exposition on all these concepts can be found in the book~\cite[Chapter 12]{gathenGerhard03}.

\section{Main Results}

We first develop the necessary tools that will be useful later to prove our main result.

\subsection{Matrix-vector multiplication complexity}\label{sec:mat-vec-complexity}
For a matrix $A \in \F^{n \times n}$, let $M_A(n)$ denote the time required to compute the \emph{matrix-vector product:} $A\vv$ for an arbitrary vector $\vv\in\F^n$. When the matrix is clear from context, we simply write $M(n)$ for convenience. The following lemmas obtain complexity for matrix-vector multiplication by multivariate Vandermonde matrices and their submatrices.

\begin{lemma}\label{lem:vandermonde-matvec}
Let $V \in \F^{N \times N}$ be a multivariate Vandermonde matrix corresponding to evaluations of $m$-variate polynomials of individual degree less than $d$ over $\F$, where $N=d^m$. Then, for any submatrix $B$ of $V$,
\[
M_B(N)=N^{1+o(1)}\poly(m,d,\log |\F|),
\]
assuming $d=\omega(1)$.
\end{lemma}

\begin{proof}

Any submatrix \(B\) of \(V\) is determined by a subset \(S \subseteq N\) of the evaluation points and a subset \(\mathcal{M}\) of the monomials from all possible $d^m$ monomial basis. For a vector \(\vv=(v_m)_{m\in \mathcal{M}}\), let
\[
f(\vx)=\sum_{m\in \mathcal{M}} v_m\,m(\vx).
\]
Then \(B\vv=f(\bm\alpha)_{\bm\alpha\in S}\). Hence, computing \(B\vv\) is precisely a restricted multivariate multipoint evaluation problem over the point set \(S\). In this setting, Bhargava et al.~\cite{BGGKU24} gave a deterministic near-linear time algorithm for multivariate multipoint evaluation over finite fields,
\[
(d^m+N)^{1+o(1)}\poly(m,d,\log |\F|).
\]
Since $N=d^m$, this simplifies to $N^{1+o(1)}\poly(m,d,\log |\F|)$. The same algorithm applies to any restriction of the evaluation points or monomials, yielding the stated bound.
\end{proof}

\begin{proposition}[\cite{BCS97}, 13.20] \label{prop:transposition principle}
Let \(A\) be an \(n \times n\) matrix and let \(\vv \in \F^n\). Suppose there exists an algebraic algorithm\footnote{Throughout, \emph{algebraic} algorithms refers to the arithmetic model, which access $\mathbb{F}_q$ only through field operations whereas \emph{non-algebraic} algorithms may additionally exploit the bit-representation of $\mathbb{F}_q$ and studied in the bit-complexity
model.} that computes the map \(\vv \mapsto A\vv\) in time \(O(n)\). Then the map \(\vv \mapsto A^{T}\vv\) can also be computed in time \(O(n)\).
\end{proposition}

The above proposition is also known as Tellegen's principle~\cite{T52}. It has been useful in the context of transposition and has been studied extensively, including in the works of Kaminski, Kirkpatrick, and Bshouty~\cite{KKB88}, Bostan, Lecerf, and Schost~\cite{BLS03}, Kedlaya and Umans~\cite{KU11}, and Kinoshita and Li~\cite{KL24}.

\subsection{Transposed multivariate multipoint evaluation}
\label{subsec:transposed-mme}

In the previous section, we described the matrix-vector multiplication complexity for multivariate Vandermonde and their submatrices, that is fast multivariate multipoint evaluation restricted to particular coordinates. In addition to that, we also need the transpose of multivariate multipoint evaluation to be fast. The transpose problem is defined as follows. Given evaluation points and a weight vector, compute all weighted power sums of the evaluation points. Note by \autoref{prop:transposition principle}, if we have an algorithm that can compute the forward $A\vv$ fast, then the transposed $A^{T}\vv$ can also be achieved in roughly same complexity, and the only criteria is the algorithm should be algebraic. Though, we know algebraic versions of multivariate multipoint evaluation due to multiple works~\cite{umans08, BGKM23}, which can lead to fast transposed multipoint evaluation, however these algorithms come with some restrictions on input parameters. 

Thus, in order to get rid of those constraints, we work with the alternative \emph{non-algebraic} version of Bhargava et al.~\cite{BGGKU24} which works over all finite fields for growing degree parameter. For that reason, the transposition principle for linear
straight-line programs do not apply to it directly. We therefore segregate 
their algorithm into different stages and analyze them and finally conclude that; although \cite{BGGKU24} is non-algebraic, the transposed multivariate multipoint evaluation runs in roughly same time as their multivariate multipoint evaluation. Note this is the same route by which Kedlaya
and Umans~\cite[Section 7]{KU11} obtain the transpose of their non-algebraic algorithm and we also follow a similar footstep.

\begin{theorem} \label{thm:transposed-mme} Let $\F$ be a finite field. Given evaluation points
$\bm a_1, \ldots, \bm a_N \in \F^m$ and a weight vector $y \in \F^N$, there is deterministic algorithm that computes all the weighted power sums
\[
  S_{\bm e}(y) \;:=\; \sum_{i=1}^{N} y_i\, \bm a_i^{\bm e}
  \qquad \text{for all } \bm e \in \intbr{d}^m,
\]
in time $(d^m + N)^{1+o(1)} \cdot \mathrm{poly}(m, d, \log|\F|)$, for all
$m \in \mathbb{N}$ and all sufficiently large $d \in \mathbb{N}$.

\medskip

Equivalently, writing $V = \big(\bm a_i^{\bm e}\big)_{i \in [N],\, \bm e \in \intbr{d}^m} \in \F^{N \times d^m}$ for the multivariate Vandermonde matrix, the map
$y \mapsto V^{\top} y$ can be computed using a deterministic algorithm that runs in same time as mentioned above.
\end{theorem}

\begin{proof}
We use the first algorithm (Algorithm $1$) of~\cite{BGGKU24}, the one valid for homogeneous polynomials but can be easily extended to non-homogeneous setting by working over all the individual homogeneous components separately. Thus, we omit the homogeneity concern, since it can have at most $dm$ many homogeneous components and each can be transposed efficiently. \medskip

Their multipoint evaluation for an $m$-variate polynomial $f$ of individual degree less than $d$, where the coefficients of $f$ denoted by $x \in \F^{d^m}$ is given as follows. The algorithm computes $Vx = \pi\big(\widetilde V\, \sigma(x)\big)$, where $\sigma$ lifts field elements to their integer representatives, $\widetilde V$ is the lifted integer Vandermonde matrix, satisfying
$\pi(\widetilde V) = V$ entry wise, and $\pi$ is the reduction back to $\F$.
 For the integer computation, the entries of $\widetilde V\, \sigma(x)$ are recovered using Chinese remaindering from their residues modulo small primes $p_1, \ldots, p_k$ whose product exceeds the apriori evaluation bound. Then the problem reduces to evaluating  $f_j(\bm a_{ij}) :\equiv \widetilde F(\widetilde{\bm a_i}) \bmod p_j$ at the reduced points over $\F_{p_j}$ for each $j$, where $\widetilde F$ is the integer lift of $f$, $f_j:= \widetilde F \mod p_j$ and $\bm a_{ij} \equiv \widetilde{a_i} \mod p_j$. 

Over $\F_{p_j}$, the evaluation is performed on the Kakeya set
$\mathcal{K} \subseteq \F_{p_j}^m$, a union of $p_j$ many product sets.  Every
point of $\F_{p_j}^m$ lies on a low-degree curve intersecting $\mathcal{K}$ at sufficiently many points. Thus, for each prime field $\F_{p_j}$, the evaluation turns into 
\[
  V_j \;=\; L_j\, E_j\, H_j,
\]
 three stages as follows. The map $H_j$ sends the coefficient vector
of $f_j$ to the stacked coefficient vectors of its Hasse derivatives of
order at most $2m$; by \cite[Proposition~3.5]{BGGKU24}, each output entry is a fixed
binomial coefficient times a single input entry. The
map $E_j$ sends these coefficient vectors to their values on $\mathcal{K}$. Since $\mathcal{K}$ is a union of product sets, we apply univariate multipoint evaluation on each of them, and which is again a linear circuit for fixed evaluation points. The map $L_j$ sends the values on $\mathcal{K}$ to the outputs $f_j(\bm a_{i,j})$; its $i$-th row is fixed linear combination of the values supported on the intersection points of $\mathcal{K}$ lying on the curve through $\bm a_{i,j}$, with Hermite interpolation weights determined by that curve. Each of $H_j, E_j,L_j$ is a linear arithmetic circuit over $\F_{p_j}$. We now consider the transpose problem. 
\[
  V^{\top} y \;=\; \pi\big(\widetilde V^{\top}\, \sigma(y)\big),
\]
so the transposed algorithm lifts $y$, computes
$\widetilde V^{\top} \sigma(y)$ modulo the primes $p_j$, recombines by the Chinese Remainder Theorem, and reduces. The outer reduction that uses the lift-reduce-CRT mechanism does not depend on the matrix it operates on, hence it applies verbatim to $\widetilde V^\top$, and the transpose inherits the outer reduction at no cost from the forward problem. In short, the lift-reduce-CRT mechanism is common to both the forward and the transpose problem, hence they do not affect the overall run time for the transpose case. A subtle point to note that, in the transpose setting, the bound for the choice of primes is slightly different though it does not affect the overall picture. Each transposed output $\sum_i \tilde y_i\, \tilde{\bm a}_i^{\bm e}$ is a sum of $N$ terms rather than $d^m$, hence bounded by $N(p-1)^{dm}$ rather than $d^m(p-1)^{dm}$, and the primes are chosen with $\prod_j p_j > N(p-1)^{dm}$. Since $N=d^m$ in our case it does not affect us. Even if they are not equal, replacing $m \log d$ by $\log N$ changes the bound on the primes negligibly, since both are at most the logarithm of the input size. So \cite[Lemma~4.3]{BGGKU24} applies as before, and the divisibility condition needed for the Kakeya construction is still met by enough primes, with the extra cost absorbed by the low-order term, i.e., $\poly(m,d,\log|\F|)$.

It now remains to transpose the inner matrix $\widetilde V_j$. Its three factors $H_j, E_j, L_j$ are linear arithmetic circuits over $\F_{p_j}$ and by Tellegen's transposition principle in~\autoref{prop:transposition principle}, each transpose $H_j^{\top}, E_j^{\top}, L_j^{\top}$ is computable with nearly same number of operations, so $V_j^{\top} = H_j^{\top} E_j^{\top} L_j^{\top}$ is computed in nearly same time as that of $V_j$.

\end{proof}

Likewise \autoref{lem:vandermonde-matvec}, the statement in \autoref{thm:transposed-mme} also holds for its submatrices bearing a similar time bound.

\begin{remark}

    As mentioned earlier, we could have used the algebraic algorithm of \cite{umans08,BGKM23} instead, which would have avoided the need to establish the transposed multivariate multipoint evaluation result in \autoref{thm:transposed-mme}. We need transposition principle in \autoref{thm: min-poly-constant-zero}. However, the algebraic approach would impose additional constraints on the underlying field, which would also be reflected in the final interpolation result. We avoid these constraints by taking the non-algebraic route.

\end{remark}

\begin{remark}[The case $N < d^m$]\label{rem:padding}
If $N < d^m$, append $d^m - N$ zero rows to $V$ and $d^m - N$ zeros to
$\bm\beta = (\beta_1,\ldots,\beta_N)$, obtaining $\widetilde V \in \F^{d^m \times d^m}$ and $\widetilde{\bm\beta} \in \F^{d^m}$. The appended equations are vacuous, so the solution set remains unchanged. But, the matrix $\widetilde V$ is no longer a Vandermonde matrix, and the padded system is a general linear system rather than an interpolation instance. However, this does not affect our algorithms which access the matrix only using matrix-vector products. For $\vy \in \F^{d^m}$ and
$(\vw, \vw') \in \F^{N} \times \F^{d^m - N}$,
\[
  \widetilde V \vy = (V\vy,\, \mathbf{0})
  \qquad\text{and}\qquad
  \widetilde V^{\top}(\vw, \vw') = V^{\top}\vw.
\]
Thus, both the products roughly cost the same using \autoref{lem:vandermonde-matvec} and \autoref{thm:transposed-mme}. On the other hand, $\rank(\widetilde V) = \rank(V) \le N$, so $\widetilde V$ is singular with co-rank at least $d^m - N$; and such singular systems are handled in \autoref{sec:singular}.
\end{remark}

\subsection{Deterministic Krylov minimal polynomial} \label{sec:det-krylov}

In this section, we prove one of the key component, that is deterministic minimal polynomial computation for the Krylov sequence.

\begin{lemma}\label{lem:compute min poly det.}
    Let \( A \in \F^{n \times n} \)  and \( \vb \in \F^n \). 
 Then, the minimal polynomial annihilating the Krylov sequence $\{A^i\vb\}_{i \in \mathbb{N}}$ can be computed in deterministic \(\widetilde{O}(n\cdot M(n)+n^2)\) time.
\end{lemma}

\begin{proof}

We use the Berlekamp--Massey algorithm to compute the minimal polynomial associated with the Krylov sequence
\[
\mathcal{K}(A,\vb)=\{A^i\vb\}_{i\in \mathbb{N}}.
\]
Since our goal is to solve the linear system, it suffices to recover the minimal polynomial of $\mathcal{K}(A,\vb)$ rather than the minimal polynomial of the matrix sequence $\{A^i\}_{i\in \mathbb{N}}$ itself. The classical Berlekamp--Massey algorithm applies to scalar recurrences. A standard randomized reduction replaces the vector sequence $\{A^i\vb\}_{i\in \mathbb{N}}$ by the scalar sequence
\[
\{\vu^TA^i\vb\}_{i\in \mathbb{N}}
\]
for a random vector $\vu$, which preserves the minimal polynomial with high probability. Instead, we avoid randomization and recover the minimal polynomial deterministically from the vector recurrence itself. Compute the vectors
\[
\{A^i\vb\}_{i=0}^{2n-1}
\]
by repeated application of $A$. For each coordinate $j\in [n]$, define the scalar sequence
\[
s_j(i)=\ve_j^TA^i\vb,
\qquad i\ge 0,
\]
where $\ve_j$ denotes the $j$-th standard basis vector. Apply the Berlekamp--Massey algorithm to the first $2n$ terms of $s_j$ to compute its minimal polynomial $\mu_j(x)$. Finally define
\[
\krymin{A}{\vb}(x)=\operatorname{lcm}(\mu_1(x),\mu_2(x),\dots,\mu_n(x)).
\]

\begin{algorithm}[H]
\caption{Deterministic Krylov Minimal Polynomial}
\textbf{Input:} $A\in \F^{n\times n}$ and vector $\vb\in \F^n$\\
\textbf{Output:} Minimal polynomial $\mu^A_{\vb}(x)$ of $\mathcal{K}(A,\vb)$
\begin{algorithmic}[1]
\State $v_0 \gets \vb$
\For{$i=1$ to $2n-1$}
    \State $v_i \gets A v_{i-1}$
\EndFor

\For{$j=1$ to $n$}
    \State Define scalar sequence 
    $s_j(i)=\ve_j^Tv_i$ for $0\le i\le 2n-1$
    \State Compute the minimal polynomial $\mu_j(x)$ of $s_j$ using Berlekamp--Massey
\EndFor

\State $\krymin{A}{\vb}(x)\gets \operatorname{lcm}(\mu_1(x),\dots,\mu_n(x))$
\State \Return $\krymin{A}{\vb}(x)$
\end{algorithmic}
\end{algorithm}

\paragraph*{Proof of correctness.}
We need to argue that the polynomial $\mu^A_{\vb}(x)$ is precisely the minimal polynomial of the Krylov sequence $\mathcal{K}(A,\vb)$.
Let $p(x)$ denote the minimal polynomial of $\mathcal{K}(A,\vb)$. Then
\[
p(A)\vb=0.
\]
Applying $\ve_j^T$ on the left gives
\[
\ve_j^Tp(A)\vb=0,
\]
and hence $p$ annihilates the scalar sequence $s_j$. Therefore, $\mu_j(x) \mid p(x)~~\text{for all }j\in [n]$. Consequently, $\krymin{A}{\vb}(x)=\operatorname{lcm}(\mu_1(x),\dots,\mu_n(x))\mid p$. Conversely, since each $\mu_j(x)$ annihilates $s_j$, we have
\[
\ve_j^T\krymin{A}{\vb}(A)\vb=0
\qquad\text{for all }j\in [n].
\]
As $\{\ve_j\}_{j=1}^n$ forms a basis of $\F^n$, it follows that $\krymin{A}{\vb}(A)\vb=0$. Hence $\krymin{A}{\vb}$ annihilates the Krylov sequence $\mathcal{K}(A,\vb)$, and therefore $p\mid \krymin{A}{\vb}$. Thus, $p=\krymin{A}{\vb}$.

\paragraph{Time complexity.}
Computing the vectors $\{A^i\vb\}_{i=0}^{2n-1}$ requires $2n-1$ matrix-vector multiplications by $A$, and hence the total cost to compute all these vectors is $O(n \cdot M(n))$ arithmetic operations. Each scalar sequence has length $2n$, and Berlekamp--Massey is applied independently to all $n$ coordinate sequences. Since one instance of the fast Berlekamp--Massey algorithm runs in $\widetilde{O}(n)$ time, the total cost of this step is $\widetilde{O}(n^2)$ operations. Finally, computing $\krymin{A}{\vb}=\operatorname{lcm}(\mu_1,\dots,\mu_n)$
involves taking the least common multiple of $n$ polynomials of degree at most $n$. Using fast polynomial arithmetic, the gcd and division with remainder of two degree-$n$ polynomials can be computed in $\widetilde{O}(n)$ time~ \cite[Chapter 9, 11]{gathenGerhard03}), and hence one lcm computation also takes $\widetilde{O}(n)$ operations. Repeating this iteratively for all $n$ polynomials requires at most $\widetilde{O}(n^2)$ arithmetic operations, since we know that the degree of the lcm is always bounded by $n$. Therefore, the minimal polynomial of the Krylov sequence $\mathcal{K}(A,\vb)$ can be computed deterministically using $\widetilde{O}(n \cdot M(n)+n^2)$ arithmetic operations.

\end{proof}

\subsection{Nonsingular system} \label{sec:nonsingular}

We first prove our results for the nonsingular system, and in the next section we do it for the singular case.

\begin{theorem} \label{thm: invertible system}
  Let \( A \in \F^{n \times n} \) be an invertible matrix and \( \vb \in \F^n \). 
Then the linear system $A\vx = \vb$ can be solved deterministically in \(\widetilde{O}(n \cdot M(n)+n^2)\) time.
\end{theorem}

\begin{proof}
Let $\krymin{A}{\vb}(x)=x^d+c_{d-1}x^{d-1}+\cdots+c_1x+c_0$, be the degree $d \leq n$ minimal polynomial of the Krylov sequence $\{A^i\vb\}_{i \in \mathbb{N}}$ computed using \autoref{lem:compute min poly det.}. Then,
$$
\krymin{A}{\vb}(A)\cdot\vb = 0,
$$
that is, $$A^d\vb + c_{d-1}A^{d-1}\vb + \cdots + c_1A\vb + c_0\vb = 0.$$

Since \(A\) is invertible, we must have \(c_0 \neq 0\). Rearranging the above equation gives 
\begin{equation} \label{eq: min. transformation}
    A\left(A^{d-1}\vb + c_{d-1}A^{d-2}\vb + \cdots + c_1\vb\right)
= -c_0\vb.
\end{equation}

Multiplying both sides by \(A^{-1}\), we obtain $$A^{d-1}\vb + c_{d-1}A^{d-2}\vb + \cdots + c_1\vb
= -c_0A^{-1}\vb.$$
Therefore,

\begin{equation}\label{eq:x}
A^{-1}\vb= \vx
=
-\frac{1}{c_0}
\left(
A^{d-1}\vb
+
c_{d-1}A^{d-2}\vb
+
\cdots
+
c_1\vb
\right).
\end{equation}  

Hence, \(A^{-1}\vb\) can be expressed as a linear combination of these set of vectors \(\vb, A\vb, A^2\vb, \ldots, A^{d-1}\vb\). The coefficients for this linear combination arising from the minimal polynomial, along with the vectors $\vb, A\vb, A^2\vb, \ldots, A^{d-1}\vb$, are already computed in \autoref{lem:compute min poly det.} and requires at most $\widetilde{O}(n\cdot M(n)+n^2)$ arithmetic operations. The linear combination in \autoref{eq:x} can be computed using at most \(O(dn)\) additional operations. Since \(d \leq n\), the total running time to solve the system is $\widetilde{O}(n\cdot M(n)+n^2)$.
\end{proof}



\begin{theorem} \label{thm:full-rank-interpolation}

Let \(N=d^m\). Given \(N\) pairs of data
\[
(\bm\alpha_1,\beta_1),\,
(\bm\alpha_2,\beta_2),\,
\ldots,\,
(\bm\alpha_N,\beta_N),
\]
where \(\bm\alpha_i\in \F^m\) and \(\beta_i\in \F\) for all \(i\in [N]\), and \(\bm\alpha_i\neq \bm\alpha_j\) for all \(i\neq j\). Let \(V\) denote the associated multivariate Vandermonde matrix corresponding to the evaluation points \(\bm\alpha_1,\ldots,\bm\alpha_N\) and monomials of individual degree less than \(d\). Suppose \(V\) is full rank. Then there exists a unique \(m\)-variate polynomial \(f(\vy)\in \F[\vy]\) with individual degree less than \(d\) such that
\[
f(\bm\alpha_i)=\beta_i
\qquad \text{for all } i\in [N],
\]
and moreover, \(f\) can be computed deterministically in \(N^{2+o(1)}\poly(m,d,\log|\F|) \) time, assuming \(d=\omega(1)\).

\end{theorem}

\begin{proof}
The proof of this theorem is a direct corollary of \autoref{thm: invertible system}. Since the multivariate Vandermonde matrix $V$ is assumed to be full rank, it remains to bound the cost of matrix-vector multiplication by $V$ in order to apply \autoref{thm: invertible system}. We already know the associated multivariate Vandermonde matrix-vector multiplication complexity $M(N)=N^{1+o(1)}\poly(m,d,\log|\F|)$ from Bhargava et al.~\cite{BGGKU24} as long as \(d=\omega(1)\). Substituting this bound into \autoref{thm: invertible system} yields the claimed interpolation algorithm run time. The existence and uniqueness of the interpolant $f$ follow from the invertibility of $V$.

\end{proof}

\subsection{Singular system} \label{sec:singular}
We first show that the \textit{constant term} of the minimal polynomial obtained using \autoref{lem:compute min poly det.} plays a crucial role in determining the interpolation strategy. We first consider the case when the \textit{constant term} is non-zero.

\begin{proposition} \label{prop: min-poly-constant-not-zero}
Let \(A \in \F^{n \times n}\) and let \(\vb \in \F^n\). Suppose that the minimal polynomial \(\krymin{A}{\vb}\) of the Krylov sequence \(\{A^i\vb\}_{i \in \mathbb{N}}\) has non-zero constant term, i.e. \(\krymin{A}{\vb}(0)\neq 0\). Then the system
\[
A\vx=\vb
\]
admits a solution \(\vx\), which can be computed deterministically in \(\widetilde{O}(n\cdot M(n)+n^2)\) time.
\end{proposition}

\begin{proof}
From \autoref{eq:x}, in order to compute a solution \(\vx\), the only requirement is that the constant term \(c_0=\krymin{A}{\vb}(0)\) is non-zero, which is already assumed. In particular, the proof does not use the invertibility of \(A\). Instead, \autoref{eq:x} shows that a solution \(\vx\) can be expressed as a linear combination of the vectors
\[
\{\vb, A\vb, A^2\vb, \ldots, A^{d-1}\vb\}.
\]

Equivalently, whenever \(\krymin{A}{\vb}(0)\neq 0\), the vector \(\vb\) lies in the column space of \(A\), and the minimal polynomial computation explicitly produces a corresponding solution vector \(\vx\). Note that the solution need not be unique.

\paragraph*{Time complexity.}
By \autoref{lem:compute min poly det.}, the minimal polynomial together with all intermediate Krylov vectors can be computed in \(\widetilde{O}(n\cdot M(n)+n^2)\) time. Using these vectors, the linear combination in \autoref{eq:x} can be computed using at most \(O(n^2)\) additional operations, since $d \leq n$. Therefore, the total running time is \(\widetilde{O}(n\cdot M(n)+n^2)\).
\end{proof}

\begin{theorem} \label{singular-nonzero-constant-inter}
Let $\F$ be a finite field and \(N=d^m\). Given \(N\) distinct data pairs
\[
(\bm\alpha_1,\beta_1),\,
(\bm\alpha_2,\beta_2),\,
\ldots,\,
(\bm\alpha_N,\beta_N),
\]
where \(\bm\alpha_i\in \F^m\) and \(\beta_i\in \F\) for all \(i\in [N]\), and \(\bm\alpha_i\neq \bm\alpha_j\) for all \(i\neq j\). Let \(V\) denote the associated $m$-variate Vandermonde matrix corresponding to the evaluation points \(\bm\alpha_1,\ldots,\bm\alpha_N\) and monomials of individual degree less than \(d\), and let
\(
\bm\beta=\bigl(\beta_1,\beta_2,\ldots,\beta_N\bigr).
\)
Suppose that the minimal polynomial of the Krylov sequence \(\{V^i\bm\beta\}_{i\in\mathbb N}\) has non-zero constant term. Then, there exists an $m$-variate polynomial \(f(\vy)\in\F[\vy]\) with individual degree less than \(d\), satisfying
\[
f(\bm\alpha_i)=\beta_i
\qquad \text{for all } i\in [N],
\]
and moreover $f$ can be computed deterministically in $N^{2+o(1)}\poly(m,d,\log|\F|)$ time, assuming \(d=\omega(1)\).
\end{theorem}

\begin{proof}
    The proof is an immediate corollary of the \autoref{prop: min-poly-constant-not-zero}, since we know the associated multivariate Vandermonde matrix-vector multiplication complexity $M(N)=N^{1+o(1)}\poly(m,d,\log|\F|)$ from Bhargava et al.~\cite{BGGKU24} as long as \(d=\omega(1)\). 
\end{proof}

\vspace{1mm}

Now, from here on, we focus on the case when \(\krymin{A}{\vb}(0)=0\), which implies that \(\corank(A) \geq 1\). We divide this into the following two subcases:
\begin{enumerate}
    \item \textbf{\(\corank(A)=O(1)\)}, 
    \item \textbf{\(\corank(A)=\omega(1)\)}.
\end{enumerate}

\subsubsection{Case 1~: \texorpdfstring{$\corank(A)=O(1)$}{}}
In this case, the aim is to prune the system by removing redundant rows and columns to reduce it to the invertible case \autoref{thm: invertible system}.

\begin{theorem} \label{thm: min-poly-constant-zero}
   Let  $A \in \F^{n \times n}$ supports fast transposed matrix-vector multiplication complexity $\widetilde{O}(M(n))$, with rank \(n-O(1)\) and $\vb \in \F^n$, where the minimal polynomial $\krymin{A}{\vb}$ of the Krylov sequence {$\{A^i\vb\}_{i \in \mathbb{N}}$} has a zero-constant term, $\krymin{A}{\vb}(0)= 0$. Then, there is a deterministic procedure to decide whether $A\vx=\vb$ is consistent and, if so, compute a solution $\vx\in\F^n$ in $\widetilde{O}(n\cdot M(n)+n^2)$ time.
\end{theorem}

\begin{proof}
Let \(\krymin{A}{\vb}(x)=x^d+c_{d-1}x^{d-1}+\cdots+c_1x+c_0\) be obtained using \autoref{lem:compute min poly det.}. Since \(c_0=0\), the right-hand side of \autoref{eq: min. transformation} is identically zero. Hence,
\begin{equation} \label{eq:dependancy-rel}
    \bm z=A^{d-1}\vb+c_{d-1}A^{d-2}\vb+\cdots+c_1\vb 
\end{equation}
gives a non-trivial linear dependence among the columns of \(A\). Indeed, \(\bm z\neq \bm 0\), since otherwise it would contradict the minimality of \(\krymin{A}{\vb}\).

Let \(j\) be the first index such that \(z_j\neq 0\). If \(z_t=0\) for all \(t\neq j\), then the above relation implies that the \(j\)-th column of \(A\) is identically zero. Otherwise, the \(j\)-th column is linearly dependent on the columns indexed by those \(t\neq j\) for which \(z_t\neq 0\). In either case, the \(j\)-th column is redundant and can be removed. Let \(A_j\) denote the matrix obtained from \(A\) by deleting its \(j\)-th column.
To remove a corresponding redundant row, consider the transposed system \(B:=A_j^T\), which is an \((n-1)\times n\) matrix. Write
\begin{equation}\label{eq:remove-column}
    B=[C\mid \bm {row_n}],
\end{equation}
where \(C\in \F^{(n-1)\times(n-1)}\) consists of the first \(n-1\) columns of \(B\), and \(\bm {row_n}\) denotes the last column\footnote{\(\bm {row_n}\) is the $n$'th row in A, whose transpose we use in \autoref{eq:remove-column} but still denote it as $\bm {row_n}$ for simplicity.}. Let \(\bm u\in \F^{n-1}\) denote the vector obtained from \(\vx\) after deleting its \(j\)-th coordinate. We now consider the system
\[
C\bm u=\bm {row_n}.
\]

We apply the deterministic Wiedemann procedure to the Krylov sequence $\krylov{C}{\bm {row_n}}$, and if the resulting minimal polynomial has non-zero constant term, then by \autoref{thm: invertible system}, the system is solvable and hence \(\bm {row_n}\) lies in the column span of \(C\); equivalently, the last row of \(A_j\) is linearly dependent on the remaining rows and can be removed. Otherwise, the minimal polynomial has  zero constant term. By \autoref{eq: min. transformation}, this yields a non-trivial linear dependence among the columns of \(C\), allowing us to remove another redundant column of \(C\), which corresponds to removing another redundant row of \(A_j\). Whenever a row of \(A_j\) is removed, the corresponding equation is removed from the reduced system. In this case, we check whether the same linear dependency relation holds on the corresponding entry of the right-hand-side vector \(\vb\). If this consistency condition fails, then the original system has no solution.

Repeating this procedure iteratively removes one redundant row and one redundant column in each step. Since the original system has rank \(n-O(1)\), we obtain a full-rank square subsystem after constantly many such iterations. Note that any solution to the reduced full-rank subsystem corresponds to a solution for the original system by assigning arbitrary values to the free variables corresponding to the deleted columns. Hence, whenever the system is consistent, the reduced system is equivalent to the original system.

\paragraph*{Time complexity.}
In each iteration we do eliminate a column and a row depending upon the the minimal polynomial computation. In order to compute the minimal polynomial of $\{A^ib\}_{i \in \N}$, it takes \(\widetilde{O}(n\cdot M(n)+n^2)\) time using \autoref{lem:compute min poly det.}. We store the intermediate vector computations ${A^ib}$ for $i \in {[2n-1]}$ for the minimal polynomial computation. These vectors can be used to take a linear combination to output the $\bm{z}$ vector, which takes $O(n^2)$ many operations and to detect the first non-zero index $j$, it takes another $O(n)$ many operations at most. This takes care of the column removal complexity. Now to remove a row, we consider the reduced transposed system $C\bm u=\bm{row_n}$. Here again we compute the minimal polynomial for $\{C^i\bm{row_n}\}_{i \in \N}$ using \autoref{lem:compute min poly det.}, which takes \(\widetilde{O}\bigl((n-1)\cdot M(n-1)+(n-1)^2\bigr)\), since the original system $A$ supports fast transposed matrix-vector multiplication. At this point, either we remove $\bm{row_n}$ from the system or a column from $C$ by computing the dependency relation as computed in \autoref{eq:dependancy-rel}, for which it takes another $O(n^2)$ many operations. 

This is just for one level of iteration and it takes total \(\widetilde{O}(n\cdot M(n)+n^2)\) many operations in total and since the algorithm runs for constantly many such iterations, the total runtime is still \(\widetilde{O}(n\cdot M(n)+n^2)\).
\end{proof}

\begin{theorem} \label{thm:interpolation-close-full-rank}

Let \(N=d^m\) and $\F$ be a finite field. Given \(N\) distinct data pairs
\[
(\bm\alpha_1,\beta_1),\,
(\bm\alpha_2,\beta_2),\,
\ldots,\,
(\bm\alpha_N,\beta_N),
\]
where \(\bm\alpha_i\in \F^m\) and \(\beta_i\in \F\) for all \(i\in [N]\), and \(\bm\alpha_i\neq \bm\alpha_j\) for all \(i\neq j\). Let \(V\) denote the associated multivariate Vandermonde matrix corresponding to the evaluation points \(\bm\alpha_1,\ldots,\bm\alpha_N\) and monomials of individual degree less than \(d\), and let
\(
\bm\beta=\bigl(\beta_1,\beta_2,\ldots,\beta_N\bigr).
\)
Suppose $\rank(V)=N-O(1)$ and the minimal polynomial of the Krylov sequence $\krylov{V}{\bm\beta}$ has a zero constant term. Then, an $m$-variate polynomial \(f(\vy)\in\F[\vy]\) with individual degree less than \(d\), satisfying
\[
f(\bm\alpha_i)=\beta_i
\qquad \text{for all } i\in [N],
\]
can be computed deterministically in $N^{2+o(1)}\poly(m,d,\log|\F|)$ time, assuming \(d=\omega(1)\), provided such an $f$ exists.
   
\end{theorem}

\begin{proof}
    The proof is an immediate corollary of \autoref{thm:transposed-mme} and \autoref{thm: min-poly-constant-zero} with $M(N)=N^{1+o(1)}\poly(m,d,\log|\F|)$ from Bhargava et al.~\cite{BGGKU24}.
\end{proof}

\subsubsection{Case 2~:~\texorpdfstring{$\corank=\omega(1)$}{}}

In this case, we cannot use the earlier approach of deleting one row and one column in each iteration to reduce the system to a full-rank one, since this may lead to a super-quadratic running time in the input size. Instead, we use a randomized approach due to Kaltofen and Saunders~\cite{KS91} to eliminate linear dependencies and determine the rank of the system.

The following lemma is a restatement of Theorem~2 from~\cite{KS91}.

\begin{lemma}[{\cite[Theorem~2]{KS91}}]\label{lem:UAL-toeplitz-transform}
Let \(A\in\mathbb{F}^{n\times n}\) be a matrix of rank \(r\), and let \(S\subseteq\mathbb{F}\). Define
\[
\widetilde{A}=UAL,
\]
where \(U\) and \(L\) are unit upper and lower triangular Toeplitz matrices, respectively, whose parameters are chosen independently and uniformly at random from \(S\). If \(\widetilde{A}_i\) denotes the leading principal \(i\times i\) submatrix of \(\widetilde{A}\), then
\[
\Pr\!\left(
\det(\widetilde{A}_i)\neq 0
\text{ for every } 1\le i\le r
\right)
\ge
1-\frac{r(r+1)}{|S|}.
\]
\end{lemma}

The above lemma transforms an arbitrary matrix \(A\) into a structured matrix whose leading principal submatrices are nonsingular up to the rank of \(A\).

The following lemma is a restatement of Lemma~2 from~\cite{KS91}. We replace the notation \(f^{AQ}\) from~\cite{KS91} by \(\mu^{AQ}\) to emphasize its interpretation as a minimal polynomial.

\begin{lemma}[{\cite[Lemma~2]{KS91}}] \label{lem:rank-comp}
Let \(A \in \mathbb{F}^{n\times n}\) have leading principal minors nonzero up to \(A_r\), where \(r < n \) is the rank of \(A\). Let \(S\subseteq \mathbb{F}^*\), and let $Q=\operatorname{diag}(\gamma_1,\ldots,\gamma_n),$ where each \(\gamma_i\) is chosen uniformly from \(S\). Then \(r=\deg(\mu^{AQ})-1\) with probability at least
\[
1-\frac{r(r+1)}{2 \,|S|}.
\]
\end{lemma}

The above lemma shows that, after multiplying a random diagonal matrix \(Q\), the degree of the minimal polynomial of \(AQ\) is equal to \(r+1\) with high probability, provided that the leading principal minors of \(A\) are nonzero up to its rank \(r\).

We now present a restatement of their Theorem 3 in ~\cite[Theorem~3]{KS91} with minor improvement in terms of randomness and a slightly modified success probability.

\begin{theorem}\label{thm:rank-comp-runtime}
Let \(A\in \F^{n\times n}\) with \(\rank(A) < n\), and let \(S\subseteq \F^*\). There exists a randomized algorithm which, using \(4n-2\) random elements from \(S\) (corresponding to the entries of \(U\), \(L\), \(Q\), and \(\vv\)), computes the rank of \(A\) using \(O(n)\) matrix-vector multiplications by \(A\) together with
\(
O(n^2\log n\log\log n)
\)
additional arithmetic operations over \(\F\). The algorithm outputs the correct rank with probability at least
\[
1-\frac{n(3n-1)}{2|S|}.
\]
\end{theorem}

\begin{proof}
By \autoref{lem:UAL-toeplitz-transform}, after the transformation
\[
\widetilde{A}=UAL,
\]
all leading principal submatrices of \(\widetilde{A}\) are nonsingular up to rank \(r=\rank(A)\) with probability at least
\[
1-\frac{r(r+1)}{|S|}.
\]

Then, by \autoref{lem:rank-comp}, $r=\deg(\mu^{\widetilde{A}Q})-1$ with probability at least
\[
1-\frac{r(r+1)}{2\,|S|}.
\]

Moreover, by the analysis given in Lemma 2 of Kaltofen and Pan~\cite{KP91},
\[
\Pr\!\left(
\mu^{\widetilde{A}Q}
=
\mu^{\widetilde{A}Q}_{\vv}
\right)
\geq
1-\frac{\deg(\mu^{\widetilde{A}Q})}{|S|}
\geq
1-\frac{n}{|S|}.
\]

Since \(r<n\), applying the union bound gives a total failure probability at most
\[
\frac{n(n-1)}{|S|}
+
\frac{n(n-1)}{2\,|S|}
+
\frac{n}{|S|}
=
\frac{n(3n-1)}{2\,|S|}.
\]

We remark that applying \(\widetilde{A}=UAL\) to a vector requires one multiplication of \(A\) by a vector together with two Toeplitz matrix-vector multiplications. Since Toeplitz matrix-vector multiplication reduces to polynomial multiplication, each application of \(\widetilde{A}\) incurs an additional \(O(n\log n\log\log n)\) arithmetic operations in \(\mathbb{F}\) due to Cantor and Kaltofen~\cite{CK91}.

Therefore, the deterministic minimal polynomial computation from \autoref{lem:compute min poly det.} requires \(O(n)\) matrix-vector multiplications and a total of \(O(n^2\log n\log\log n)\) arithmetic operations in \(\mathbb{F}\). Finally, replacing the randomized scalar Wiedemann projection from~\cite{KS91} with the deterministic minimal polynomial computation removes the need for one random vector, reducing the total number of random field elements from \(5n-2\) to \(4n-2\).
\end{proof}

\begin{theorem} \label{thm: all-rank-min-poly-constant-zero}
   Let  $A \in \F^{n \times n}$ and $\vb \in \F^n$, where the minimal polynomial $\mu^A_{\vb}$ of the Krylov sequence $\krylov{A}{\vb}$ has a zero-constant term, i.e. $\mu^A_{\vb}(0)= 0$, then there exists a randomized procedure that determines whether the system $A\vx=\vb$ is consistent and if so, computes a solution $\vx \in \F^n$ in \(\widetilde{O}(n \cdot M(n)+n^2)\) time. The algorithm uses $4n-2$ many random elements from $\F^*$ and succeeds with probability at least 
\[1-\frac{n(3n-1)}{2\,(|\F|-1)},\]
moreover, it never outputs an incorrect solution. 
\end{theorem}

\begin{proof}

Let $\widetilde{A}=UAL$ be the preconditioned matrix from \autoref{lem:UAL-toeplitz-transform}. Since $U,L$ are nonsingular, solving $A\vx=\vb$
reduces to solving
\[
\widetilde{A}{\bm {{\bm {\widetilde{x}}}}}=U\vb,
\]
because if ${\bm {\widetilde{x}}}$ satisfies $\widetilde{A}{\bm {\widetilde{x}}}=U\vb$, then
\(\vx=L{\bm {\widetilde{x}}}.\)

Assume
\[
\widetilde{A}=
\begin{pmatrix}
\widetilde{A}_r & B\\
C & D
\end{pmatrix},
\]
where $\operatorname{rank}(\widetilde{A})=r$ and the leading principal block $\widetilde{A}_r$ is nonsingular. Consider the equivalent matrix
\[
\begin{pmatrix}
\widetilde{A}_r & B\\
C & D
\end{pmatrix}
\begin{pmatrix}
I & -\widetilde{A}_r^{-1}B\\
0 & I
\end{pmatrix}
=
\begin{pmatrix}
\widetilde{A}_r & 0\\
C & D-C\widetilde{A}_r^{-1}B
\end{pmatrix}.
\]
Since the right factor is invertible, both matrices have the same rank. As $\operatorname{rank}(\widetilde{A})=r$ and $\widetilde{A}_r$ is nonsingular, it follows that
\[
D=C\widetilde{A}_r^{-1}B.
\]

Now write
\[
{\bm {\widetilde{x}}}=
\begin{pmatrix}
\bm {x_1}\\
\bm{x_2}
\end{pmatrix},
\qquad
U\vb=
\begin{pmatrix}
\bm{b_1}\\
\bm{b_2}
\end{pmatrix},
\]
with $\bm{x_1},\bm{b_1}\in \F^r$. Then
\[
\widetilde{A}{\bm {\widetilde{x}}}=
\begin{pmatrix}
\widetilde{A}_r \bm{x_1}+B\bm{x_2}\\
C\bm{x_1}+C\widetilde{A}_r^{-1}B\bm{x_2}
\end{pmatrix}.
\]
Hence, $\widetilde{A}{\bm {\widetilde{x}}}=U\vb$ iff
\[
\widetilde{A}_r \bm{x_1}=\bm{b_1}-B\bm{x_2}
\qquad\text{and}\qquad
\bm{b_2}=C\widetilde{A}_r^{-1}\bm{b_1}.
\]
The latter condition is independent of $\bm{x_2}$ and is necessary and sufficient for consistency.

We solve it for the particular choice of $\bm{x_2}=\bm{0}$,
\begin{equation} \label{eq:comp-y}
    \widetilde{A}_r \bm{y}=\bm{b_1}.
\end{equation}

Apply deterministic Wiedemann's algorithm from \autoref{thm: invertible system} to solve for the unique vector $\bm y \in \F^r$, since $\widetilde{A}_r$ is nonsingular. Although the algorithm requires matrix-vector multiplication by $\widetilde{A}_r$, we never form $\widetilde{A}_r$ explicitly. Let
\[
D_r=\operatorname{diag}(\underbrace{1,\dots,1}_{r},0,\dots,0)\in \F^{n\times n}.
\]
For $\bm z\in \F^r$, identify $\bm z$ with $(z_1,\dots,z_r,0,\dots,0)^T\in \F^n$. Then,
\[
D_r\widetilde{A}D_r \bm z
=
\begin{pmatrix}
\widetilde{A}_r \bm z\\
\bm 0
\end{pmatrix}.
\]
Hence every matrix-vector multiplication by $\widetilde{A}_r$ required to compute the minimal polynomial of the Krylov sequence $\krylov{\widetilde{A}_r}{\bm b_1}$
 is realized using $\widetilde{A}=UAL$ together with left and right multiplication by $D_r$. Consequently, \autoref{thm: invertible system} computes a unique solution $\bm y$ to \autoref{eq:comp-y}. Now define
\[
\bm{\widetilde{x}}=
\begin{pmatrix}
\bm{y}\\
\bm{0}
\end{pmatrix}.
\]

Then,
\[
\widetilde{A}\bm{\widetilde{x}}=
\begin{pmatrix}
\widetilde{A}_r \bm y\\
C\bm y
\end{pmatrix}
=
\begin{pmatrix}
\bm{ b_1}\\
C\widetilde{A}_r^{-1}\bm{b_1}
\end{pmatrix}.
\]
Therefore,
\[
\widetilde{A}\bm{\widetilde{x}}=U\vb
\iff
\bm{b_2}=C\widetilde{A}_r^{-1}{\bm b_1}.
\]
Thus, consistency is verified implicitly by one additional multiplication with $\widetilde{A}$ and comparison with $U\vb$. Finally, once we have $\bm{\widetilde{x}}$, the corresponding solution to the original system is given by
\[
\bm x=L\bm{\widetilde{x}}.
\]

The probability that the algorithm fails is precisely the failure probability of the rank computation procedure from \autoref{thm:rank-comp-runtime}. The algorithm never returns an incorrect solution. Since $\widetilde{A}=UAL$ and $U$ is nonsingular, the equality $\widetilde{A}\bm{\widetilde{x}}=U\vb$ holds if and only if $A(L\bm{\widetilde{x}})=\vb$. Hence, the output vector $\vx=L\bm{\widetilde{x}}$ solves the system with certainty, irrespective of the randomized rank computation procedure. The only possible error is that the test fails on a consistent system due to failure in rank computation step, in which case the algorithm reports inconsistency. Thus the algorithm is Las Vegas on positive output.
\paragraph{Time complexity.}

We compute the rank $r$ of the preconditioned matrix $\widetilde{A}=UAL$ from \autoref{thm:rank-comp-runtime}. The matrix $\widetilde{A}$ is never formed explicitly; rather, every multiplication by $\widetilde{A}$ is implemented as
\[
\vv \mapsto U(A(L\vv)),
\]
where multiplication by the Toeplitz matrices $U$ and $L$ is carried out via fast polynomial multiplication. Consequently, the rank computation requires $O(n)$ matrix-vector multiplications by $A$ together with $O(n^2\log n\log\log n)$ additional arithmetic operations as given by \autoref{thm:rank-comp-runtime}. Computing the transformed right-hand side $Ub$ similarly takes $O(n\log n\log\log n)$ operations.

To solve \autoref{eq:comp-y}, we apply \autoref{thm: invertible system} to $\widetilde{A}_r$. The matrix $\widetilde{A}_r$ is also never formed explicitly. Instead, 
matrix-vector multiplication by $\widetilde{A}_r$ is realized through the operator
\[
\bm z \mapsto D_r\widetilde{A}D_r \bm z.
\]
Since $\widetilde{A}_r$ is an $r\times r$ nonsingular matrix, \autoref{thm: invertible system} requires $O(r)$ matrix-vector multiplications by $A$. Applying \(U\), \(L\), and \(D_r\) costs at most \(O(rn\log n\log\log n)\) arithmetic operations in total. Recovering the solution vector \(y\) from the minimal polynomial relation requires another \(O(r^2)\) operations. Finally, append $n-r$ zeros to $\bm y$ to form $\bm{\widetilde{x}}$ and recover a solution to the original system as $\bm x=L\bm{\widetilde{x}}$, which takes $O(n\log n\log\log n)$ additional operations. Since $r \le n$, the overall complexity is bounded by $\widetilde{O}(n\cdot M(n)+n^2),$ that is near-quadratic time.
\end{proof}

\begin{theorem} \label{thm:interpolation-all-rank-min-poly-constant-zero}
    Let \(N=d^m\) and $\F$ be a finite field. Given \(N\) distinct data pairs
\[
(\bm\alpha_1,\beta_1),\,
(\bm\alpha_2,\beta_2),\,
\ldots,\,
(\bm\alpha_N,\beta_N),
\]
where \(\bm\alpha_i\in \F^m\) and \(\beta_i\in \F\) for all \(i\in [N]\), and \(\bm\alpha_i\neq \bm\alpha_j\) for all \(i\neq j\). Let \(V\) denote the associated multivariate Vandermonde matrix corresponding to the evaluation points \(\bm\alpha_1,\ldots,\bm\alpha_N\) and monomials of individual degree less than \(d\), and let
\(
\bm\beta=\bigl(\beta_1,\beta_2,\ldots,\beta_N\bigr).
\)
Suppose the minimal polynomial of the Krylov sequence $\krylov{V}{\bm\beta}$  has a zero constant term. Then, there is a randomized procedure to determine if there exists an $m$-variate polynomial \(f(\vy)\in\F[\vy]\) with individual degree less than \(d\), satisfying
\[
f(\bm\alpha_i)=\beta_i
\qquad \text{for all } i\in [N],
\]
and if such an $f$ exists, it can be computed in $N^{2+o(1)}\poly(m,d,\log|\F|)$ time, upon assumption \(d=\omega(1)\), using $4N-2$ random field elements from $\F^*$ with success probability at least $$1-\frac{N(3N-1)}{2\,(|\F|-1)},$$ furthermore, it never outputs an incorrect solution.

\end{theorem}

\begin{proof}

This follows directly from \autoref{thm: all-rank-min-poly-constant-zero} by
substituting $M(N)= N^{1+o(1)}\poly(m,d,\log |\F|)$, using the result of
Bhargava et al.~\cite{BGGKU24}. By \autoref{thm: all-rank-min-poly-constant-zero} the algorithm never outputs an incorrect solution, so the only possible error is to declare a consistent system inconsistent. Taking $|\F| = \Omega(N^2)$ makes the failure probability small. If $|\F|$ is small, one works over a field extension. When the original system is consistent, independent repetition decreases the failure probability exponentially.

\end{proof}

Note that the above theorem also generalizes the constant co-rank case in \autoref{thm:interpolation-close-full-rank}. In the regime where the rank is guaranteed to be sufficiently close to full rank, the system can be solved deterministically without requiring any randomness, even though the exact rank is not known apriori.

\paragraph{Acknowledgment}
The author is grateful to Mrinal Kumar for several insightful discussions and suggestions throughout the course of this work including presentation of the draft. The author acknowledges M.K. for providing the initial ideas for this project, in particular for the nonsingular system, although the approach was randomized rather than deterministic~(see \autoref{thm: invertible system}). The author is also grateful to the reviewers for their valuable suggestions to improve the draft quality. \medskip

\paragraph{AI Statement} Generative AI has been used for the literature study and in particular for the \textit{Related Work} in Section 3. It has also been used for copy-editing and to improve the overall presentation.

\bibliography{interpolation}

@article {BGGKU24,
    AUTHOR = {Bhargava, Vishwas and Ghosh, Sumanta and Guo, Zeyu and Kumar,
              Mrinal and Umans, Chris},
     TITLE = {Fast multivariate multipoint evaluation over all finite
              fields},
   JOURNAL = {J. ACM},
  FJOURNAL = {Journal of the ACM},
    VOLUME = {71},
      YEAR = {2024},
    NUMBER = {3},
     PAGES = {Art. 22, 32},
      ISSN = {0004-5411,1557-735X},
   MRCLASS = {68Q11 (11T71 12-08)},
  MRNUMBER = {4763252},
MRREVIEWER = {Dorothy\ Bollman},
       DOI = {10.1145/3652025},
       
}

@article {BGKM23,
    AUTHOR = {Bhargava, Vishwas and Ghosh, Sumanta and Kumar, Mrinal and
              Mohapatra, Chandra Kanta},
     TITLE = {Fast, algebraic multivariate multipoint evaluation in small
              characteristic and applications},
   JOURNAL = {J. ACM},
  FJOURNAL = {Journal of the ACM},
    VOLUME = {70},
      YEAR = {2023},
    NUMBER = {6},
     PAGES = {Art. 42, 46},
      ISSN = {0004-5411,1557-735X},
   MRCLASS = {68Q06 (68P05)},
  MRNUMBER = {4679508},
       DOI = {10.1145/3625226},
      
}

@article {KKB88,
    AUTHOR = {Kaminski, Michael and Kirkpatrick, David G. and Bshouty, Nader
              H.},
     TITLE = {Addition requirements for matrix and transposed matrix
              products},
   JOURNAL = {J. Algorithms},
  FJOURNAL = {Journal of Algorithms},
    VOLUME = {9},
      YEAR = {1988},
    NUMBER = {3},
     PAGES = {354--364},
      ISSN = {0196-6774},
   MRCLASS = {68Q25 (15-04 65F05)},
  MRNUMBER = {955144},
MRREVIEWER = {D.\ Z.\ Spicer},
       DOI = {10.1016/0196-6774(88)90026-0},
      
}

@article{T52,
  title={A general network theorem, with applications},
  author={Tellegen, Bernard DH},
  journal={Philips Res Rep},
  volume={7},
  pages={256--269},
  year={1952}
}

@ARTICLE{w86,
  author={Wiedemann, D.},
  journal={IEEE Transactions on Information Theory}, 
  title={Solving sparse linear equations over finite fields}, 
  year={1986},
  volume={32},
  number={1},
  pages={54-62},
  doi={10.1109/TIT.1986.1057137}}

@article {KU11,
    AUTHOR = {Kedlaya, Kiran S. and Umans, Christopher},
     TITLE = {Fast polynomial factorization and modular composition},
   JOURNAL = {SIAM J. Comput.},
  FJOURNAL = {SIAM Journal on Computing},
    VOLUME = {40},
      YEAR = {2011},
    NUMBER = {6},
     PAGES = {1767--1802},
      ISSN = {0097-5397,1095-7111},
   MRCLASS = {11Y16 (11Y05 12E05 68W30 68W40)},
  MRNUMBER = {2863194},
MRREVIEWER = {Samuel\ S.\ Wagstaff, Jr.},
       DOI = {10.1137/08073408X},
      
}

@incollection {KL24,
    AUTHOR = {Kinoshita, Yasunori and Li, Baitian},
     TITLE = {Power series composition in near-linear time},
 BOOKTITLE = {2024 {IEEE} 65th {A}nnual {S}ymposium on {F}oundations of
              {C}omputer {S}cience---{FOCS} 2024},
     PAGES = {2180--2185},
 PUBLISHER = {IEEE Computer Soc., Los Alamitos, CA},
      YEAR = {[2024] \copyright 2024},
      ISBN = {979-8-3315-1674-1},
   MRCLASS = {68W30 (13F25)},
  MRNUMBER = {4849324},
       DOI = {10.1109/FOCS61266.2024.00127},
       URL = {https://doi.org/10.1109/FOCS61266.2024.00127},
}

@book {BCS97,
    AUTHOR = {B\"urgisser, Peter and Clausen, Michael and Shokrollahi, M.
              Amin},
     TITLE = {Algebraic complexity theory},
    SERIES = {Grundlehren der mathematischen Wissenschaften [Fundamental
              Principles of Mathematical Sciences]},
    VOLUME = {315},
      NOTE = {With the collaboration of Thomas Lickteig},
 PUBLISHER = {Springer-Verlag, Berlin},
      YEAR = {1997},
     PAGES = {xxiv+618},
      ISBN = {3-540-60582-7},
   MRCLASS = {68-02 (12Y05 65Y20 68Q05 68Q15 68Q25 68Q40)},
  MRNUMBER = {1440179},
MRREVIEWER = {Alexander\ I.\ Barvinok},
       DOI = {10.1007/978-3-662-03338-8},
       URL = {https://doi.org/10.1007/978-3-662-03338-8},
}

@inproceedings{BLS03,
author = {Bostan, A. and Lecerf, G. and Schost, \'{E}.},
title = {Tellegen's principle into practice},
year = {2003},
isbn = {1581136412},
publisher = {Association for Computing Machinery},
address = {New York, NY, USA},
doi = {10.1145/860854.860870},
booktitle = {Proceedings of the 2003 International Symposium on Symbolic and Algebraic Computation},
pages = {37–44},
numpages = {8},
location = {Philadelphia, PA, USA},
series = {ISSAC '03}
}

@InProceedings{FKW25,
  author =	{Fischer, Nick and Kallmayer, Melvin and Wennmann, Leo},
  title =	{{A Simple Algorithm for Trimmed Multipoint Evaluation}},
  booktitle =	{33rd Annual European Symposium on Algorithms (ESA 2025)},
  pages =	{89:1--89:11},
  series =	{Leibniz International Proceedings in Informatics (LIPIcs)},
  ISBN =	{978-3-95977-395-9},
  ISSN =	{1868-8969},
  year =	{2025},
  volume =	{351},
  editor =	{Benoit, Anne and Kaplan, Haim and Wild, Sebastian and Herman, Grzegorz},
  publisher =	{Schloss Dagstuhl -- Leibniz-Zentrum f{\"u}r Informatik},
  address =	{Dagstuhl, Germany},
  URL =		{https://drops.dagstuhl.de/entities/document/10.4230/LIPIcs.ESA.2025.89},
  URN =		{urn:nbn:de:0030-drops-245574},
  doi =		{10.4230/LIPIcs.ESA.2025.89}
}

@incollection {KS91,
    AUTHOR = {Kaltofen, Erich and Saunders, B. David},
     TITLE = {On {W}iedemann's method of solving sparse linear systems},
 BOOKTITLE = {Applied algebra, algebraic algorithms and error-correcting
              codes ({N}ew {O}rleans, {LA}, 1991)},
    SERIES = {Lecture Notes in Comput. Sci.},
    VOLUME = {539},
     PAGES = {29--38},
 PUBLISHER = {Springer, Berlin},
      YEAR = {1991},
      ISBN = {3-540-54522-0},
   MRCLASS = {68Q40},
  MRNUMBER = {1229306},
       DOI = {10.1007/3-540-54522-0\_93}
}

@inproceedings{KP91,
author = {Kaltofen, Erich and Pan, Victor},
title = {Processor efficient parallel solution of linear systems over an abstract field},
year = {1991},
isbn = {0897914384},
publisher = {Association for Computing Machinery},
address = {New York, NY, USA},
doi = {10.1145/113379.113396},
booktitle = {Proceedings of the Third Annual ACM Symposium on Parallel Algorithms and Architectures},
pages = {180–191},
numpages = {12},
location = {Hilton Head, South Carolina, USA},
series = {SPAA '91}
}

@article {CK91,
    AUTHOR = {Cantor, David G. and Kaltofen, Erich},
     TITLE = {On fast multiplication of polynomials over arbitrary algebras},
   JOURNAL = {Acta Inform.},
  FJOURNAL = {Acta Informatica},
    VOLUME = {28},
      YEAR = {1991},
    NUMBER = {7},
     PAGES = {693--701},
      ISSN = {0001-5903,1432-0525},
   MRCLASS = {68Q40},
  MRNUMBER = {1129288},
       DOI = {10.1007/BF01178683},
}

@article {GS00,
    AUTHOR = {Gasca, Mariano and Sauer, Thomas},
     TITLE = {Polynomial interpolation in several variables},
      NOTE = {Multivariate polynomial interpolation},
   JOURNAL = {Adv. Comput. Math.},
  FJOURNAL = {Advances in Computational Mathematics},
    VOLUME = {12},
      YEAR = {2000},
    NUMBER = {4},
     PAGES = {377--410},
      ISSN = {1019-7168,1572-9044},
   MRCLASS = {41A15 (41-02 41A25 41A63 65D10)},
  MRNUMBER = {1768957},
MRREVIEWER = {Peter\ Alfeld},
       DOI = {10.1023/A:1018981505752},
}

@article {ChungYao1977,
    AUTHOR = {Chung, K. C. and Yao, T. H.},
     TITLE = {On lattices admitting unique {L}agrange interpolations},
   JOURNAL = {SIAM J. Numer. Anal.},
  FJOURNAL = {SIAM Journal on Numerical Analysis},
    VOLUME = {14},
      YEAR = {1977},
    NUMBER = {4},
     PAGES = {735--743},
      ISSN = {0036-1429},
   MRCLASS = {41A05 (65D05)},
  MRNUMBER = {445158},
MRREVIEWER = {P.\ G.\ Ciarlet},
       DOI = {10.1137/0714050},
}

@article {deBoorRon1990,
    AUTHOR = {de Boor, Carl and Ron, Amos},
     TITLE = {On multivariate polynomial interpolation},
   JOURNAL = {Constr. Approx.},
  FJOURNAL = {Constructive Approximation. An International Journal for
              Approximations and Expansions},
    VOLUME = {6},
      YEAR = {1990},
    NUMBER = {3},
     PAGES = {287--302},
      ISSN = {0176-4276,1432-0940},
   MRCLASS = {41A05 (41A10)},
  MRNUMBER = {1054756},
MRREVIEWER = {S.\ D.\ Riemenschneider},
       DOI = {10.1007/BF01890412},
       
}

@inproceedings{Fid72,
author = {Fiduccia, Charles M.},
title = {Polynomial evaluation via the division algorithm the fast Fourier transform revisited},
year = {1972},
isbn = {9781450374576},
publisher = {Association for Computing Machinery},
address = {New York, NY, USA},
doi = {10.1145/800152.804900},
booktitle = {Proceedings of the Fourth Annual ACM Symposium on Theory of Computing},
pages = {88–93},
numpages = {6},
location = {Denver, Colorado, USA},
series = {STOC '72}
}

@article {BM74,
    AUTHOR = {Borodin, A. and Moenck, R.},
     TITLE = {Fast modular transforms},
   JOURNAL = {J. Comput. System Sci.},
  FJOURNAL = {Journal of Computer and System Sciences},
    VOLUME = {8},
      YEAR = {1974},
     PAGES = {366--386},
      ISSN = {0022-0000},
   MRCLASS = {68A20 (12-04)},
  MRNUMBER = {371144},
MRREVIEWER = {J.\ Morgenstern},
       DOI = {10.1016/S0022-0000(74)80029-2},
       URL = {https://doi.org/10.1016/S0022-0000(74)80029-2},
}

@book {Lorentz1992,
    AUTHOR = {Lorentz, Rudolph A.},
     TITLE = {Multivariate {B}irkhoff interpolation},
    SERIES = {Lecture Notes in Mathematics},
    VOLUME = {1516},
 PUBLISHER = {Springer-Verlag, Berlin},
      YEAR = {1992},
     PAGES = {x+192},
      ISBN = {3-540-55870-5},
   MRCLASS = {41A05 (41-02 41A63 65D05 65N30)},
  MRNUMBER = {1222648},
MRREVIEWER = {Petru\ P.\ Blaga},
       DOI = {10.1007/BFb0088788},
      
}

@article {DynFloater2014,
    AUTHOR = {Dyn, Nira and Floater, Michael S.},
     TITLE = {Multivariate polynomial interpolation on lower sets},
   JOURNAL = {J. Approx. Theory},
  FJOURNAL = {Journal of Approximation Theory},
    VOLUME = {177},
      YEAR = {2014},
     PAGES = {34--42},
      ISSN = {0021-9045,1096-0430},
   MRCLASS = {41A05 (41A10 41A63)},
  MRNUMBER = {3128402},
MRREVIEWER = {Stefano\ De Marchi},
       DOI = {10.1016/j.jat.2013.09.008},
       
}

@misc{Hofmann2025,
 author        = {Hofmann, Phil-Alexander and  Michael Hecht},
 title         = {Accelerating Multivariate Newton Interpolation in Downward Closed Polynomial Spaces},
 year          = {2025},
 eprint        = {2505.14909},
  archivePrefix = {arXiv},
    doi={https://doi.org/10.48550/arXiv.2505.14909},
}

@incollection {Zippel1979,
    AUTHOR = {Zippel, Richard},
     TITLE = {Probabilistic algorithms for sparse polynomials},
 BOOKTITLE = {Symbolic and algebraic computation ({EUROSAM} '79, {I}nternat.
              {S}ympos., {M}arseille, 1979)},
    SERIES = {Lecture Notes in Comput. Sci.},
    VOLUME = {72},
     PAGES = {216--226},
 PUBLISHER = {Springer, Berlin-New York},
      YEAR = {1979},
      ISBN = {3-540-09519-5},
   MRCLASS = {68C20},
  MRNUMBER = {575692},
  DOI={10.1007/3-540-09519-5_73},
}

@inproceedings{BenOrTiwari1988,
author = {Ben-Or, Michael and Tiwari, Prasoon},
title = {A deterministic algorithm for sparse multivariate polynomial interpolation},
year = {1988},
isbn = {0897912640},
publisher = {Association for Computing Machinery},
address = {New York, NY, USA},

doi = {10.1145/62212.62241},
booktitle = {Proceedings of the Twentieth Annual ACM Symposium on Theory of Computing},
pages = {301–309},
numpages = {9},
location = {Chicago, Illinois, USA},
series = {STOC '88}
}

@article {Zippel1990,
    AUTHOR = {Zippel, Richard},
     TITLE = {Interpolating polynomials from their values},
   JOURNAL = {J. Symbolic Comput.},
    VOLUME = {9},
      YEAR = {1990},
    NUMBER = {3},
     PAGES = {375--403},
      ISSN = {0747-7171,1095-855X},
   MRCLASS = {68Q40 (12Y05 65D05)},
  MRNUMBER = {1056633},
       DOI = {10.1016/S0747-7171(08)80018-1},
}

@techreport{GiebrechtLabahnLee2002,
  author      = {Giesbrecht, Mark and Labahn, George and Lee, Wen-shin},
  title       = {On the equivalence between {Prony's} and {Ben-Or/Tiwari's} methods},
  institution = {University of Waterloo},
  year        = {2002},
  url={https://win.uantwerpen.be/~wlee/pub/CS2002-23.pdf},
}

@incollection {KaltofenYagati1988,
    AUTHOR = {Kaltofen, Erich and Yagati, Lakshman},
     TITLE = {Improved sparse multivariate polynomial interpolation
              algorithms},
 BOOKTITLE = {Symbolic and algebraic computation ({R}ome, 1988)},
    SERIES = {Lecture Notes in Comput. Sci.},
    VOLUME = {358},
     PAGES = {467--474},
 PUBLISHER = {Springer, Berlin},
      YEAR = {1989},
      ISBN = {3-540-51084-2},
   MRCLASS = {65D05},
  MRNUMBER = {1034749},
       DOI = {10.1007/3-540-51084-2\_44},
     
}

@article {KaltofenLee2003,
    AUTHOR = {Kaltofen, Erich and Lee, Wen-shin},
     TITLE = {Early termination in sparse interpolation algorithms},
      NOTE = {International Symposium on Symbolic and Algebraic Computation
              (ISSAC'2002) (Lille)},
   JOURNAL = {J. Symbolic Comput.},
  FJOURNAL = {Journal of Symbolic Computation},
    VOLUME = {36},
      YEAR = {2003},
    NUMBER = {3-4},
     PAGES = {365--400},
      ISSN = {0747-7171,1095-855X},
   MRCLASS = {13P05 (68W30)},
  MRNUMBER = {2004034},
MRREVIEWER = {Andrei\ Kelarev},
       DOI = {10.1016/S0747-7171(03)00088-9},

     
}

@article {BorodinTiwari1991,
    AUTHOR = {Borodin, Allan and Tiwari, Prasoon},
     TITLE = {On the decidability of sparse univariate polynomial
              interpolation},
   JOURNAL = {Comput. Complexity},
  FJOURNAL = {Computational Complexity},
    VOLUME = {1},
      YEAR = {1991},
    NUMBER = {1},
     PAGES = {67--90},
      ISSN = {1016-3328,1420-8954},
   MRCLASS = {12L05 (03B25 12Y05 68Q15 68Q40)},
  MRNUMBER = {1113535},
MRREVIEWER = {John\ Pedersen},
       DOI = {10.1007/BF01200058},
      
}

@incollection {GrigorievKarpinski1993,
    AUTHOR = {Grigoriev, Dima and Karpinski, Marek},
     TITLE = {A zero-test and an interpolation algorithm for the shifted
              sparse polynomials},
 BOOKTITLE = {Applied algebra, algebraic algorithms and error-correcting
              codes ({S}an {J}uan, {PR}, 1993)},
    SERIES = {Lecture Notes in Comput. Sci.},
    VOLUME = {673},
     PAGES = {162--169},
 PUBLISHER = {Springer, Berlin},
      YEAR = {1993},
      ISBN = {3-540-56686-4},
   MRCLASS = {12Y05 (68Q20)},
  MRNUMBER = {1251977},
MRREVIEWER = {Maurice\ Mignotte},
       DOI = {10.1007/3-540-56686-4\_41},
     
}

@inproceedings {BlaserJindal2014,
    AUTHOR = {Bl\"aser, Markus and Jindal, Gorav},
     TITLE = {A new deterministic algorithm for sparse multivariate
              polynomial interpolation},
 BOOKTITLE = {I{SSAC} 2014---{P}roceedings of the 39th {I}nternational
              {S}ymposium on {S}ymbolic and {A}lgebraic {C}omputation},
     PAGES = {51--58},
 PUBLISHER = {ACM, New York},
      YEAR = {2014},
      ISBN = {978-1-4503-2501-1},
   MRCLASS = {68W30 (41A05 65D05)},
  MRNUMBER = {3239908},
       DOI = {10.1145/2608628.2608648},
      
}

@article {HuangGao2020,
    AUTHOR = {Huang, Qiao-Long and Gao, Xiao-Shan},
     TITLE = {Faster interpolation algorithms for sparse multivariate
              polynomials given by straight-line programs},
   JOURNAL = {J. Symbolic Comput.},
  FJOURNAL = {Journal of Symbolic Computation},
    VOLUME = {101},
      YEAR = {2020},
     PAGES = {367--386},
      ISSN = {0747-7171,1095-855X},
   MRCLASS = {65D05 (65C05 68W40)},
  MRNUMBER = {4109723},
       DOI = {10.1016/j.jsc.2019.10.005},     
}

@article {Huang2023,
    AUTHOR = {Huang, Qiao-Long},
     TITLE = {Sparse polynomial interpolation based on derivatives},
   JOURNAL = {J. Symbolic Comput.},
  FJOURNAL = {Journal of Symbolic Computation},
    VOLUME = {114},
      YEAR = {2023},
     PAGES = {359--375},
      ISSN = {0747-7171,1095-855X},
   MRCLASS = {68Q04 (11T06 12E20)},
  MRNUMBER = {4446195},
MRREVIEWER = {\'Isabel\ Pirsic},
       DOI = {10.1016/j.jsc.2022.06.002},
      
}

@incollection {HuangGao2017,
    AUTHOR = {Huang, Qiao-Long and Gao, Xiao-Shan},
     TITLE = {Sparse polynomial interpolation with finitely many values for
              the coefficients},
 BOOKTITLE = {Computer algebra in scientific computing},
    SERIES = {Lecture Notes in Comput. Sci.},
    VOLUME = {10490},
     PAGES = {196--209},
 PUBLISHER = {Springer, Cham},
      YEAR = {2017},
      ISBN = {978-3-319-66320-3; 978-3-319-66319-7},
   MRCLASS = {65D05 (41A05 41A10)},
  MRNUMBER = {3706565},
       DOI = {10.1007/978-3-319-66320-3\_15},
       
}

@article {Werther1994,
    AUTHOR = {Werther, Kai},
     TITLE = {The complexity of sparse polynomial interpolation over finite
              fields},
   JOURNAL = {Appl. Algebra Engrg. Comm. Comput.},
  FJOURNAL = {Applicable Algebra in Engineering, Communication and
              Computing},
    VOLUME = {5},
      YEAR = {1994},
    NUMBER = {2},
     PAGES = {91--103},
      ISSN = {0938-1279,1432-0622},
   MRCLASS = {12Y05 (11T06 68Q40)},
  MRNUMBER = {1264459},
MRREVIEWER = {Jeffrey\ O.\ Shallit},
       DOI = {10.1007/BF01438278},
   
}

@article {GrigorievKarpinskiSinger1990,
    AUTHOR = {Grigoriev, Dima Yu. and Karpinski, Marek and Singer, Michael
              F.},
     TITLE = {Fast parallel algorithms for sparse multivariate polynomial
              interpolation over finite fields},
   JOURNAL = {SIAM J. Comput.},
  FJOURNAL = {SIAM Journal on Computing},
    VOLUME = {19},
      YEAR = {1990},
    NUMBER = {6},
     PAGES = {1059--1063},
      ISSN = {0097-5397},
   MRCLASS = {11Y16 (11T06 68Q22)},
  MRNUMBER = {1069099},
MRREVIEWER = {Joachim\ von zur Gathen},
       DOI = {10.1137/0219073},
    
}

@misc{Grenet2023,
  author = {Grenet, Bruno},
  title  = {Sparse polynomial interpolation},
  note   = {Presentation at {ICALP} 2023 Algebraic Complexity Theory workshop},
  year   = {2023},
  url={https://www.dcs.warwick.ac.uk/~u2270030/act2023icalp/slides/ICALP2023ACT_grenet.pdf},
}

@article {SauerXu1995b,
    AUTHOR = {Sauer, Thomas and Xu, Yuan},
     TITLE = {On multivariate {H}ermite interpolation},
   JOURNAL = {Adv. Comput. Math.},
  FJOURNAL = {Advances in Computational Mathematics},
    VOLUME = {4},
      YEAR = {1995},
    NUMBER = {3},
     PAGES = {207--259},
      ISSN = {1019-7168,1572-9044},
   MRCLASS = {41A63 (65D05)},
  MRNUMBER = {1357718},
MRREVIEWER = {Peter\ Alfeld},
       DOI = {10.1007/BF03177515},
       
}

@article {vanderHoevenSchost2013,
    AUTHOR = {van der Hoeven, Joris and Schost, \'Eric},
     TITLE = {Multi-point evaluation in higher dimensions},
   JOURNAL = {Appl. Algebra Engrg. Comm. Comput.},
  FJOURNAL = {Applicable Algebra in Engineering, Communication and
              Computing},
    VOLUME = {24},
      YEAR = {2013},
    NUMBER = {1},
     PAGES = {37--52},
      ISSN = {0938-1279,1432-0622},
   MRCLASS = {68W30},
  MRNUMBER = {3011318},
       DOI = {10.1007/s00200-012-0179-3},
       URL = {https://doi.org/10.1007/s00200-012-0179-3},
}

@incollection {MollerBuchberger1982,
    AUTHOR = {M\"oller, H. M. and Buchberger, B.},
     TITLE = {The construction of multivariate polynomials with preassigned
              zeros},
 BOOKTITLE = {Computer algebra ({M}arseille, 1982)},
    SERIES = {Lecture Notes in Comput. Sci.},
    VOLUME = {144},
     PAGES = {24--31},
 PUBLISHER = {Springer, Berlin-New York},
      YEAR = {1982},
      ISBN = {3-540-11607-9},
   MRCLASS = {12-04 (12E05 68C20)},
  MRNUMBER = {680050},
MRREVIEWER = {Daniel\ Lazard},
}

@article {MarinariMollerMora1993,
    AUTHOR = {Marinari, M. G. and M\"oller, H. M. and Mora, T.},
     TITLE = {Gr\"obner bases of ideals defined by functionals with an
              application to ideals of projective points},
   JOURNAL = {Appl. Algebra Engrg. Comm. Comput.},
  FJOURNAL = {Applicable Algebra in Engineering, Communication and
              Computing},
    VOLUME = {4},
      YEAR = {1993},
    NUMBER = {2},
     PAGES = {103--145},
      ISSN = {0938-1279,1432-0622},
   MRCLASS = {13P10 (15A04 68Q25)},
  MRNUMBER = {1223853},
MRREVIEWER = {Joachim\ Apel},
       DOI = {10.1007/BF01386834},
    
}

@article{PAN1994,
title = {Simple Multivariate Polynomial Multiplication},
journal = {Journal of Symbolic Computation},
volume = {18},
number = {3},
pages = {183-186},
year = {1994},
issn = {0747-7171},
doi = {https://doi.org/10.1006/jsco.1994.1042},
author = {Victor Y. Pan}
}

@article {ZilicVranesic2002,
    AUTHOR = {Zilic, Zeljko and Vranesic, Zvonko G.},
     TITLE = {A deterministic multivariate interpolation algorithm for small
              finite fields},
   JOURNAL = {IEEE Trans. Comput.},
  FJOURNAL = {Institute of Electrical and Electronics Engineers.
              Transactions on Computers},
    VOLUME = {51},
      YEAR = {2002},
    NUMBER = {9},
     PAGES = {1100--1105},
      ISSN = {0018-9340,1557-9956},
   MRCLASS = {94B35},
  MRNUMBER = {1922596},
       DOI = {10.1109/TC.2002.1032628},
       
}

@article {N19,
    AUTHOR = {Neidinger, Richard D.},
     TITLE = {Multivariate polynomial interpolation in {N}ewton forms},
   JOURNAL = {SIAM Rev.},
  FJOURNAL = {SIAM Review},
    VOLUME = {61},
      YEAR = {2019},
    NUMBER = {2},
     PAGES = {361--381},
      ISSN = {1095-7200,0036-1445},
   MRCLASS = {41A05 (41A10 41A63 65D05)},
  MRNUMBER = {3947286},
MRREVIEWER = {Stefano\ De Marchi},
       DOI = {10.1137/17M1124188},
       
}

@article {ArnoldGiesroche16,
    AUTHOR = {Arnold, Andrew and Giesbrecht, Mark and Roche, Daniel S.},
     TITLE = {Faster sparse multivariate polynomial interpolation of
              straight-line programs},
   JOURNAL = {J. Symbolic Comput.},
  FJOURNAL = {Journal of Symbolic Computation},
    VOLUME = {75},
      YEAR = {2016},
     PAGES = {4--24},
      ISSN = {0747-7171,1095-855X},
   MRCLASS = {68W30 (41A05 65D05 68W20)},
  MRNUMBER = {3451329},
       DOI = {10.1016/j.jsc.2015.11.005},
      
}

@article {GargSchost09,
    AUTHOR = {Garg, Sanchit and Schost, \'Eric},
     TITLE = {Interpolation of polynomials given by straight-line programs},
   JOURNAL = {Theoret. Comput. Sci.},
  FJOURNAL = {Theoretical Computer Science},
    VOLUME = {410},
      YEAR = {2009},
    NUMBER = {27-29},
     PAGES = {2659--2662},
      ISSN = {0304-3975,1879-2294},
   MRCLASS = {68W30},
  MRNUMBER = {2531108},
       DOI = {10.1016/j.tcs.2009.03.030},
     
}

@article {Man95,
    AUTHOR = {Mansour, Yishay},
     TITLE = {Randomized interpolation and approximation of sparse
              polynomials},
   JOURNAL = {SIAM J. Comput.},
  FJOURNAL = {SIAM Journal on Computing},
    VOLUME = {24},
      YEAR = {1995},
    NUMBER = {2},
     PAGES = {357--368},
      ISSN = {0097-5397},
   MRCLASS = {65D05 (41A10)},
  MRNUMBER = {1320215},
MRREVIEWER = {Antonio\ L\'opez-Carmona},
       DOI = {10.1137/S0097539792239291},
     
}

@inproceedings{ArnoldRoche14,
author = {Arnold, Andrew and Roche, Daniel S.},
title = {Multivariate sparse interpolation using randomized Kronecker substitutions},
year = {2014},
isbn = {9781450325011},
publisher = {Association for Computing Machinery},
address = {New York, NY, USA},

doi = {10.1145/2608628.2608674},
booktitle = {Proceedings of the 39th International Symposium on Symbolic and Algebraic Computation},
pages = {35–42},
numpages = {8},
location = {Kobe, Japan},
series = {ISSAC '14}
}

@article {HoevenLec25,
    AUTHOR = {van der Hoeven, Joris and Lecerf, Gr\'egoire},
     TITLE = {Fast interpolation of multivariate polynomials with sparse
              exponents},
   JOURNAL = {J. Complexity},
  FJOURNAL = {Journal of Complexity},
    VOLUME = {87},
      YEAR = {2025},
     PAGES = {Paper No. 101922, 29},
      ISSN = {0885-064X,1090-2708},
   MRCLASS = {68W30 (11Y16 14Q20 65D05)},
  MRNUMBER = {4843897},
       DOI = {10.1016/j.jco.2024.101922},
       
}

@inproceedings{GiorgiGrePerrRoc22,
author = {Giorgi, Pascal and Grenet, Bruno and Perret du Cray, Armelle and Roche, Daniel S.},
title = {Sparse Polynomial Interpolation and Division in Soft-linear Time},
year = {2022},
isbn = {9781450386883},
publisher = {Association for Computing Machinery},
address = {New York, NY, USA},
doi = {10.1145/3476446.3536173},
booktitle = {Proceedings of the 2022 International Symposium on Symbolic and Algebraic Computation},
pages = {459–468},
numpages = {10},
location = {Villeneuve-d'Ascq, France},
series = {ISSAC '22}
}

@article {CKLMSS,
    AUTHOR = {Chaugule, Prasad and Kumar, Mrinal and Limaye, Nutan and
              Mohapatra, Chandra Kanta and She, Adrian and Srinivasan,
              Srikanth},
     TITLE = {Schur polynomials do not have small formulas if the
              determinant does not},
   JOURNAL = {Comput. Complexity},
  FJOURNAL = {Computational Complexity},
    VOLUME = {32},
      YEAR = {2023},
    NUMBER = {1},
     PAGES = {Paper No. 3, 47},
      ISSN = {1016-3328,1420-8954},
   MRCLASS = {68Q06 (05E05)},
  MRNUMBER = {4591410},
MRREVIEWER = {Martin\ C.\ Cooper},
       DOI = {10.1007/s00037-023-00236-x},
       
}

@book {gathenGerhard03,
    AUTHOR = {von zur Gathen, Joachim and Gerhard, J\"urgen},
     TITLE = {Modern computer algebra},
   EDITION = {Second},
 PUBLISHER = {Cambridge University Press, Cambridge},
      YEAR = {2003},
     PAGES = {xiv+785},
      ISBN = {0-521-82646-2},
   MRCLASS = {68W30 (11Y16 68-01 68-02)},
  MRNUMBER = {2001757},
MRREVIEWER = {Jeffrey\ O.\ Shallit},
}

@inproceedings{LiuNguyenYang26,
author = {Yiping Liu and Hoai-An Nguyen and Junzhao Yang},
title = {Numerical Linear Algebra in Linear Space},
booktitle = {Proceedings of the 2026 Annual ACM-SIAM Symposium on Discrete Algorithms (SODA)},
year={2026},
doi = {10.1137/1.9781611978971.230},
}

@incollection {umans08,
    AUTHOR = {Umans, Christopher},
     TITLE = {Fast polynomial factorization and modular composition in small
              characteristic},
 BOOKTITLE = {S{TOC}'08},
     PAGES = {481--490},
 PUBLISHER = {ACM, New York},
      YEAR = {2008},
      ISBN = {978-1-60558-047-0},
   MRCLASS = {68W30 (68W20)},
  MRNUMBER = {2582906},
       DOI = {10.1145/1374376.1374445},
       URL = {https://doi.org/10.1145/1374376.1374445},
}

\end{document}